\documentclass[11pt,english]{article}
\usepackage[T1]{fontenc}
\usepackage{color}
\usepackage{amsmath}
\usepackage{amssymb}
\usepackage{setspace}
\usepackage{esint}
\usepackage{multirow}
\makeatletter
\usepackage{amsthm}
\usepackage{graphicx}
\usepackage{amsfonts}
\usepackage{slashed}
\usepackage{afterpage}
\usepackage{cite}
\usepackage[all]{xy}

\usepackage[margin=1in]{geometry}

\usepackage{babel}

\makeatother

\usepackage{babel}
\newenvironment{rem}[1][{}]{\smallbreak \noindent  {\bf Remark #1}\small }

\def\RR{{\mathbb R}}
\def\CC{{\mathbb C}}
\def\HH{{\mathbb H}}
\def\NN{{\mathbb N}}
\def\OO{{\mathbb O}}

\def\starad{\mathrm{ad}^*}
\def\starAd{\mathrm{Ad}^*}
\def\ad{\mathrm{ad}}
\def\Ad{\mathrm{Ad}}

\def\bra{\langle}
\def\ket{\rangle}

\newtheorem{theorem}{Theorem}
\newtheorem{definition}{Definition}
\newtheorem{proposition}{Proposition}
\newtheorem{lemma}{Lemma}

\begin{document}

	\title{Particle models from special Jordan backgrounds and spectral triples}
	\author{F. Besnard$^{a,}$\footnote{fabien.besnard@epf.fr}{~,~}S. Farnsworth$^{b,}$\footnote{shane.farnsworth@aei.mpg.de}}
	
	\maketitle
	 \begin{center}
   $^a$EPF \'Ecole d'ing\'enieurs, France\\
   $^b$Max Planck Institute for Gravitational Physics (Albert Einstein Institute), Germany.
	\end{center} 
 
	\begin{abstract}
		 We put forward a definition for spectral triples and algebraic backgrounds based on Jordan coordinate algebras. We also propose natural and gauge-invariant bosonic configuration spaces of fluctuated Dirac operators and compute them for general, almost-associative, Jordan, coordinate algebras. We emphasize that the theory so obtained is not equivalent with  usual associative noncommutative geometry, even when the coordinate algebra is the self-adjoint part of a $C^*$-algebra. In particular, in the Jordan case, the gauge fields are always unimodular, thus curing a long-standing problem in noncommutative geometry.
	\end{abstract}

\section{Introduction}

%Particle models from Jordan spectral geometry
%Jordan spectral geometry
%Jordan spectral geometry and particle theory
%Particle physics and Jordan spectral geometry
%Particle theories from Jordan spectral triples
%Jordan spectral triples and particle theory

One of the major achievements of 20\textsuperscript{th} century fundamental physics % of the second half of the  20\textsuperscript{th} century 
    was the discovery that elementary particles are subject to internal symmetries, i.e. symmetries which are not associated with the four dimensions of spacetime. These symmetries are described by gauge theory, which %descends from the older idea of Kaluza and Klein that the spacetime manifold has invisible extra dimensions. One difference between gauge theory and Kaluza-Klein theory is that in gauge theory the extended manifold takes the form of a vector bundle, a form which is preserved by the symmetries. Passing from Kaluza-Klein to gauge theory thus implies a very stringent restriction on the symmetries. Another difference between the two approaches is that Kaluza-Klein theory uses the metric as the variable, while gauge theory uses the gauge connection. Of course it must be added that Kaluza-Klein theory cannot account for the bosonic fields of the Standard Model in any natural way (many additional fields are unavoidable) and suffers from other problems like the instability of the compact dimensions. On the other hand, gauge theory was able to 
successfully incorporates  the non-gravitational fundamental forces among the elementary particles, with mass generation accounted for by the presence of scalar fields (also known as Higgs fields). While this description is 
highly successful from a phenomenological perspective,  the situation remains somewhat unsatisfactory, as already at the classical level gravitational and non-gravitational forces are described by different kinds of theories. A common theme has been the development of a quantum theory that unifies gravity with the other fundamental forces associated with internal gauge symmetries.

 In the 1990's noncommutative geometry (NCG) emerged as an intriguing unifying framework\footnote{See \cite{Chamseddine:2020} for a historical survey.}. 
 NCG provides an approach very similar to Kaluza-Klein theory, but in which the extra-dimensions     are noncommutative, which means they are described by a (finite-dimensional) noncommutative  $C^*$-algebra. In short, noncommutative geometry allows one to replace the usual internal Riemannian space of a Kaluza-Klein theory with a more general kind of geometry that avoids all of the stability issues that would otherwise  arise through compactification. Within the NCG framework all of the fundamental forces are unified, with the internal forces corresponding to `gravity' in the internal geometry, with the  Dirac operator taken to be the dynamical variable. Moreover, Higgs fields do not have to be put in by hand, and have a natural interpretation as  the finite components of the Dirac operator. It must be added that NCG  is not only   a  reformulation, but is more restrictive than gauge theory and is predictive\cite{Chamseddine:2012fk,Dungen,BesBrou}. 
 
 Despite the many successful features of the NCG approach to unification, it has suffered from some technical problems (see the introduction of \cite{AB1} for a list). By now most of these problems have  admitted at least partial solutions. For instance, only a Euclidean version of NCG was available at first,   but some aspects have now been extended to general signature. However, the only action principle which is known to work for those extensions does not include gravity\footnote{Some progress on this front may be close, see \cite{VietMichal}.}. Hence the \emph{signature problem} may be said to be partially resolved. On the other hand, the so called \emph{unimodularity problem} remains a puzzle within the NCG framework. In brief, this problem is that  one has to remove a $U(1)$-factor by hand through a \emph{ad hoc} `{unimodularity condition}' in order to recover the correct gauge group for the Standard Model in NCG~\cite{Chamseddine:2007oz}. As was argued in~\cite{Farnsworth:2020ozj}, the unimodularity problem is automatically resolved by replacing associative noncommutative coordinate algebras by Jordan algebras, hence moving from noncommutative to nonassociative geometry of Jordan type\footnote{This fact was first noted, although not published, by Latham Boyle.}. {  One of the key motivations} for the present work is to back this claim with a more formal proof.

 This paper focuses on the description of gauge theories as nonassociative geometries of Jordan type. Several works have already been devoted to non-associative geometry and in particular Jordan geometry (\cite{Farnsworth:2014vva}, \cite{Boyle:2020}, \cite{Boyle:2014wba}), with  one of the original motivations being  a   resolution to the so called `fermion doubling' problem, that arises in the associative NCG construction of the standard model of particle  physics\cite{LMMS}, but which is neatly circumvented in the Jordan setting~\cite{Boyle:2020} (see also Barrett's related solution~\cite{Barrett2007}, and that proposed by Connes~\cite{Connes:2006qv}). Further motivation arises from the geometric construction of $E_6$ grand unified theories, where  a related approach based on the exceptional Jordan  algebra has emerged (\cite{DV:2016}, \cite{DV:2019}, see also \cite{krasnov,boyle2020}) and drawn legitimate attention~\footnote{See also \cite{Furey:2018drh}  for an unrelated but interesting approach to finding standard model representation spaces using nonassociative algebras.}. However, despite this attention, the detailed framework of NCG, including the axioms of  spectral triples, the definition of the fields, and so on,  have not yet been completely formalized in the Jordan setting. The main goal of this paper is to propose a complete translation from the NCG framework. In doing so, we will observe that the Jordan formalism admits a more natural representation of the symmetries of the algebra inside the triple. Moreover, we will show in more explicit detail that the unimodularity problem disappears. Our primary focus will be  on special  Jordan coordinate algebras, although we will also briefly discuss the more general case at the end of the paper. It is important to observe that Jordan geometry is not equivalent to associative NCG even in the special (\emph{i.e.} non-exceptional) case, as the resolution of the unimodularity and fermion doubling problems show.

Let us summarize more precisely the content of this paper. In section \ref{SectionAssoc} we   briefly recall the axioms of  (associative) spectral triples, and how particle models such as the Standard Model can be defined using them. In particular we  recall that unitary elements of the algebra give rise to symmetries of the model thanks to a map $\Upsilon$. Hence a factor $M_3(\CC)$ in the algebra generates a $U(3)$ group of symmetries while we would need a $SU(3)$ for the Standard Model,  thus creating the unimodularity problem. In section \ref{secAB} we  recall the notion of algebraic background, which one of the authors has proposed as an analogue of the background manifold for a noncommutative Kaluza-Klein theory. Algebraic backgrounds permit one to define a configuration space in which the Dirac operator (i.e. the dynamic variable capturing all Bosonic degrees of freedom), lives. In section \ref{secConfigFluct} we recall the definition of this space and explain how fluctuated Dirac operators form a subspace which is invariant under non-gravitational symmetries. Hence, the space of fluctuated Dirac operators can be used as the bosonic configuration space for a particle model in which gravity is ignored. This will end the reminders about associative noncommutative geometry. 

In order to deal with Jordan geometry, we first recall a few facts on Jordan algebras, Jordan-Banach algebras and their derivations in section \ref{prelimJB}. In particular we recall the notions of associative and multiplicative representations. The first kind of representation only exists for special  algebras, while the latter is a generalization which is available for all Jordan algebras.  In section \ref{secJBT}, we propose a definition of Jordan spectral triples and Jordan algebraic backgrounds in the non-exceptional case. 
In particular we observe that a real special Jordan triple or background is always equipped not only with an associative representation $\pi$ and an opposite  action $\pi^o$ (i.e. a bi-representation), but, thanks to the order $0$ condition, also with a multiplicative representation $S=\frac{1}{2}(\pi+\pi^o)$, which turns out to play a key role. It should also to be noted that the module of Jordan 1-forms (replacing noncommutative 1-forms), is both a Lie module and a Jordan module. 

Section \ref{fluctuationsection} is devoted to fluctuations and symmetries, a key  part of the present paper. In particular we show that, at least for \emph{almost-associative algebras}, i.e. algebras $A=\mathcal{C}(M,A_F)$  of continuous functions with values in a finite-dimensional Jordan algebra $A_F$\footnote{Defined in analogy with `almost-commutative' coordinate algebras of noncommutative geometry.}, the Lie algebra of inner derivations  is isomorphic through the derivative of  $\Upsilon$     to $[S(A),S(A)]$. This means that inner derivations, and by exponentiation inner automorphisms, are directly represented on the Hilbert space of the Jordan triple, while in the associative case one had to use unitary elements. This is the reason why Jordan triples naturally implement unimodularity: the elements of  $[S(A),S(A)]$  are traceless. In section \ref{ordersymmetries} we give an interpretation of the order $1$ condition which is peculiar to the Jordan setting and allows us to define a generalized form of fluctuations. Section \ref{AAsec} is a specialization to the case of `almost-associative' Jordan algebras.  We give a motivated definition for the module of Jordan 1-forms in this case, and most importantly we compute the space of fluctuated Dirac operators of an almost-associative triple and put forward a condition (weaker that the order $1$ condition) under which it is automorphism invariant. Sections \ref{BFmodel} and \ref{PSmodel} are applications to the B-L extension of the standard model studied in \cite{Boyle:2020} and the Pati-Salam model {  respectively}. 
%The second model is recovered without using extra terms in the fluctuations, and the first one is shown to differ from the Standard Model. One of the problem is the way the $U(1)$ part of the gauge group is implemented in the Jordan setting. [gotta be a bit careful here, because in that paper we also showed that the model differs from the standard model]
In section \ref{secbeyond} {  we discuss the extension of  the results of this paper to the general setting (i.e. including exceptional Jordan coordinate algebras) by} directly using a multiplicative representation $\rho$, without assuming that it has the form $S=\frac{1}{2}\pi$, {  where $\pi$ is an associative representation}. This generalization applies in particular to the exceptional case where no associative representation exists. We show that inner derivations are still directly implemented on the Hilbert space, and we discuss the representation of 1-forms and the fluctuation of Dirac operators in this more general setting.

%, However, we have no clue how to define 1-forms, and because of this, we cannot generalize algebraic backgrounds nor the order-one condition or its weaker versions.

In the whole paper, the ``Standard Model'' we discuss includes 3 generations of right-handed neutrinos and the see-saw mechanism. We will use the symbol $\dagger$ to denote the involution in a $*$-algebra. In particular, if $T$ is an operator we will write $T^\dagger$ for its adjoint. The symbol $*$ will be used for complex conjugates.  We will generally use the symbol $A$ for a Jordan algebra and $\mathcal{A}$ for an associative one.

\section{Associative spectral triples and the Standard Model}\label{SectionAssoc}

In this section we recall some basics on spectral triples in the usual setting involving $C^*$-algebras, as well as the less familiar notion of algebraic backgrounds and what they are good for. We will talk about \emph{associative} spectral triples in order to distinguish them from the Jordan spectral triples to be defined below. A { real, even, } associative, spectral triple\footnote{We give the definition of a real and even spectral triple. This is the only case we consider in this paper.} is a  multiplet $\mathcal{T}=(\mathcal{A},H,\pi,D,\chi,J)$, where $\mathcal{A}$ is a unital $*$-algebra, $H$ is a Hilbert space, $\pi$ a $*$-representation of $\mathcal{A}$ on $H$ which we suppose to be faithful, $D$ a formally selfadjoint operator on $H$, $\chi$ a bounded selfadjoint operator, and $J$ an anti-unitary operator such that:

\begin{itemize}
\item the norm closure of $\pi(\mathcal{A})$ is a $C^*$-algebra,
\item for each $a\in \mathcal{A}$, $\pi(a)$ is in the domain of the derivation $[D,.]$,
\item $\chi$ commutes with $\pi(\mathcal{A})$ and anticommutes with $D$, while $J$ commutes with $D$,
\item $\chi^2=1$,
\item for each $a,b\in \mathcal{A}$, $[\pi(a),\pi(b)^o]=0$ ($C_0$),
\item $\forall a,b\in \mathcal{A}$, $[\pi(a)^o,[D,\pi(b)]]=0$,\ ($C_1$),
\item $J\chi=\epsilon'' \chi J$, $J^2=\epsilon$, where $\epsilon,\epsilon''=\pm 1$,
\end{itemize}
where we have used the very handy notation $T^o:=JT^\dagger J^{-1}$. Note, that since we are dealing with the even case we can, without loss of generality, choose the real structure such that $JD = DJ$\cite{Farnsworth:2016qbp}.

The idea behind these axioms is that  the elements of $\mathcal{A}$ represent virtual differentiable functions on a noncommutative manifold. We will often refer to the ``manifold case'', which is the paradigmatic example where $\mathcal{A}$ is really the algebra of differentiable functions and $D$ the canonical Dirac operator on a spin manifold\footnote{For the precise role of each object, consult \cite{costarica} or \cite{Walterbook}.}. In order to state a reconstruction theorem \cite{Connes:reconstruction}, additional conditions  having to do with the topology of the   manifold under consideration would be required, but we will not need them in this paper. Note that the order 1 condition ($C_1$), plays an essential role  in the manifold case since it ensures that $D$ is a first-order differential operator. One gains access to more general geometries, however, by dropping $C_1$. As many of these more general geometries are physically interesting, including the B-L extension given in \cite{Bes21}\footnote{See also the B-L extension given in~\cite{Boyle:2020}, which maintains $C_1$ but relies on outer automorphisms to obtain an extended standard model gauge group.} or the Pati-Salam model \cite{Chamseddine:2013uq}, our discussion will focus primarily on the more general setting. We also need to recall the definition of \emph{noncommutative $1$-forms}, since they will soon turn out to play a key role. A noncommutative $1$-form $\omega$ is a finite sum  

\begin{equation}
\omega=\sum_i \pi(a_i)[D,\pi(b_i)]\label{def1forms}
\end{equation}

with $a_i,b_i\in \mathcal{A}$. The space of all such $1$-forms is an associative $\mathcal{A}$-bimodule denoted $\Omega^1_D{\mathcal A}$. In the manifold case a noncommutive $1$-form is a field with values in the space generated by gamma-matrices. It can be understood as a vector field seen as an operator on spinors.

Let us now briefly explain how particle models are  defined using spectral triples. There are three fundamental insights :

\begin{description}
\item[a.] The Dirac operator is a replacement for the manifold metric.
\item[b.]\label{secidea} The symmetry group of a gauge theory coupled to gravity, which is the semi-direct  product of the diffeomorphism group with the gauge symmetry group, can be interpreted (roughly) as the automorphism group of the $*$-algebra $\mathcal{A} = \mathcal{C}^\infty(M)\otimes \mathcal{A}_F$, where $\mathcal{A}_F$ is a finite-dimensional algebra.%\footnote{More correctly it is the automorphism group of the Eilenberg algebra corresponding to the representation of $A$ on the Hilbert space $H$ of the spectral triple\cite{}.}.
\item[c.] The bosonic part of the Standard Model coupled to gravity can thus be rewritten as a unified theory with the Dirac operator as the variable, defined on an almost-commutative manifold, i.e. the tensor product of the canonical triple over a manifold with a finite-dimensional noncommutative triple. As a bonus, we shall recover the Higgs fields as the finite component of the Dirac operator.
\end{description}

Let us now look at the technical implementation of these beautiful ideas. A first difficulty arises with the second one. The gauge group of the Standard Model is $U(1)\times SU(2)\times SU(3)/(\mathbb{Z}_2\times \mathbb{Z}_3)$. Now $SU(2)\times SU(3)/(\mathbb{Z}_2\times \mathbb{Z}_3)$ is the automorphism group of the algebra $M_2(\CC)\oplus M_3(\CC)$, but we cannot make $U(1)$ appear in this way. This problem can be solved by considering unitary elements of the algebra instead of automorphisms (this will be justified below). Then we can consider  the algebra $\mathcal{A}_{SM}=\CC\oplus \HH\oplus M_3(\CC)$ which has unitary group $U(1)\times SU(2)\times U(3)$. We see that  we still need to reduce $U(3)$ to $SU(3)$, which is dealt with by the \emph{unimodularity condition} to be explained below.

\begin{rem}
Let us observe that the algebra $\mathcal{A}_{SM}$ deviates a little bit from the $C^*$-paradigm since it is not a complex but a real algebra. There is a theory of real $C^*$-algebras \cite{Li}, but it has some subtleties, in particular concerning the Gelfand-Naimark duality which is at the root of the noncommutative geometry program. This difficulty, which is seldom emphasized, might be an additional clue hinting towards Jordan algebras, which are better behaved than real $C^*$-algebras in many respects.
\end{rem}

Let us now turn to the gauge fields. Gauge fields appear when we ``fluctuate'' the manifold Dirac operator $D_M$ with an inner automorphism $u$ of the algebra, i.e. we replace $D_M$ with $\pi(u)D_M\pi(u)^{-1}$ where

\begin{equation}
u : x\mapsto u(x)\label{umap}
\end{equation} 

is a smooth map from $M$ to the unitary group of the finite algebra. Rewriting $\pi(u)D_M\pi(u)^{-1}$ as $D_M+\pi(u)[D_M,\pi(u)^{-1}]$, we see the term $\pi(u)[D_M,\pi(u)^{-1}]$ appear, which is a `pure gauge' field with values in the Lie algebra of the gauge group. The idea with the replacement $D_M\mapsto \pi(u)D_M\pi(u)^{-1}$ is that $\pi(u)$ acts as an automorphism of the spectral triple, and that $\pi(u)D_M\pi(u)^{-1}$ is just as good a Dirac operator as $D_M$ is. The problem is that we haven't yet specified what  a spectral triple automorphism is exactly and what  \emph{a} Dirac operator is, as opposed to \emph{the} Dirac operator we start with. Clearly, we need a structure in which the Dirac operator is allowed to vary, which is not the case with a spectral triple. Let us use the following definitions.

\begin{definition}
A \emph{pre-spectral triple} $B = (\mathcal{A},H,\pi,\chi,J)$ is the same thing as a spectral triple with the Dirac operator removed. A  pre-spectral triple automorphism    is a unitary operator $U$ such that:
\begin{enumerate}
\item $U\pi(\mathcal{A})U^{-1}=\pi(\mathcal{A})$,
\item $U$ commutes with $J$ and $\chi$.
\end{enumerate}
 A spectral triple automorphism is a pre-spectral triple automorphism which also commutes with the Dirac operator.
\end{definition}

 Note that pre-spectral triple are called \emph{fermion spaces} in \cite{Barrett2015}.

Let $\mathrm{Aut} (\mathcal{A})$ be the group of   $*$-automorphisms of $\mathcal{A}$  and $\mathrm{Aut}(B)$ the group of pre-spectral triple automorphisms of $B$.  There is a natural homomorphism $\mathrm{Aut}(B)\rightarrow \mathrm{Aut} (\mathcal{A})$ which sends $U$ to the restriction of $\mathrm{Ad}_U$  on $\pi(\mathcal{A})$, and thus defines an automorphism of $\mathcal{A}$ since $\pi$ is faithful. What we would need in order to fully implement  insight {\bf b} above is a section of this homomorphism. This  does not exist in general, however we can always lift the unitary elements of $\mathcal{A}$ thanks to the following map: for any $a \in \mathcal{A}$  define 

\begin{equation}
\Upsilon(a)=\pi(a)J\pi(a)J^{-1}
\end{equation}

Thanks to $C_0$, the map $\Upsilon$   satisfies for all $a,b\in \mathcal{A}$

\begin{enumerate}
\item  $[J,\Upsilon(a)]=0$,
\item $\Upsilon(ab)=\Upsilon(a)\Upsilon(b)$,
\item $\Upsilon(a^\dagger)=\Upsilon(a)^\dagger$.
\end{enumerate}
 
In particular $\Upsilon$ defines a group homomorphism from $U(\mathcal{A})$, the unitary group of $\mathcal{A}$,  to $U(H)$, the group of unitary operators on $H$. Moreover, we immediately see using $C_0$ that for $u\in U(\mathcal{A})$, $\Upsilon(u)$ is a pre-spectral triple automorphism such that $\mathrm{Ad}_{\Upsilon(u)}=\mathrm{Ad}_u$. We thus have the following commutative diagram:
\begin{equation}
\xymatrix{U(\mathcal{A}) \ar[r]^{\Upsilon}\ar[d]_{\mathrm{Id}} &  \mathrm{Aut}(B)\ar[d]_{\mathrm{Ad}}  \cr
U(\mathcal{A})\ar[r]_{\mathrm{Ad}} & \mathrm{Aut}_{\rm Inn}(\mathcal{A})
}
\end{equation}
where $\mathrm{Aut}_{\rm Inn}(\mathcal{A})$ is the group of inner $*$-automorphisms of $\mathcal{A}$.
\begin{rem}
As explained by one of the authors in \cite{Boyle:2014wba}, a beautiful and compact way of understanding pre-spectral triples and their  automorphisms is in terms of  `Eilenberg' algebras. That is, a pre-spectral triple can be thought of as a graded, involutive algebra $B$\footnote{We will use the same notation `$B$' when describing a  pre-spectral triple as an Eilenberg algebra, as the two descriptions are  equivalent.} with vector space $\mathcal{A} \oplus H$. The algebra product between elements in $\mathcal{A}$ is provided by the usual product on $\mathcal{A}$, while the product between elements in $\mathcal{A}$ and $H$ is provided by the representation of $\mathcal{A}$ on $H$, and the product between elements in $H$ is always equal to zero. The involution on elements in $\mathcal{A}$ is simply inherited from $\mathcal{A}$ itself, while the involution on elements in $H$ is provided by $J$. The operator $\chi$ provides a grading on $H$. As indicated in the above diagram, the unitary elements $\Upsilon(u)$ for $u\in \mathcal{A}$, then satisfy the defining properties of automorphisms on the `Eilenberg' algebra $B$, and in particular:
\begin{align}
\Upsilon(u)(\pi(a)h) &= \pi(\Ad_{u}a)\Upsilon(u)(h)
\end{align}
for all $a\in \mathcal{A}$ and $h\in H$. The observation that Upsilon should be viewed as an automorphism on $B$ is key to generalizing to the non-special Jordan setting, as well as  to the non-associative setting more generally~\cite{Farnsworth:2014vva}.

% The map Upsilon has been present implicitly in NCG since the appearance of real structure \cite{Connes : NCG and reality paper}. It has been emphasized and its properties analyzed in \cite{Walter's book}, however its signifance lies even further that the traditional approach to NCG. 

%This observation is key to understanding not only the exceptional case, but also non-associative geometries more generally. 

 %For instance, as explained by one etc." Also, since Upsilon might be key to understanding the non-special case, I would add something at the end, like 

\end{rem}

Let us see how $\Upsilon$ works in the example of the Standard Model. Since a unitary element of $\mathcal{C}(M,\mathcal{A}_{\mathrm{SM}})$ is a function with values in $U(\mathcal{A}_{\mathrm{SM}})$, we only need to examine the finite part of the model.  The finite Hilbert space is then 

\begin{equation}
H_F=\CC^2_{\mathrm{weak}}\otimes\CC^4_{\rm colour}\otimes \CC^4_\chi\otimes \CC^N\label{SMHilbertSpace}
\end{equation}

where $\CC^2_{\mathrm{weak}}$ is the weak isospin space, the \emph{colour space} $\CC^4_{\rm color}$ is $\CC^4=\CC_{\ell}\oplus\CC^3_{rgb}$, where ``leptonicity'' is seen as a fourth colour, and  the \emph{chirality space} $\CC_\chi$  is generated by the letters $R,L, \bar R,\bar L$. The last factor $\CC^N$ is the generation space, where $N$ is usually equals to $3$.  The  finite representation $\pi_F : \mathcal{A}_{\rm SM}\rightarrow\mathrm{End}(H_F)$ is 

\begin{equation}
\pi_F(\lambda,q,m)=[ \begin{pmatrix}
\lambda&0\cr 0&\lambda^*
\end{pmatrix}\otimes 1_4,q\otimes 1_4,\lambda 1_2\oplus 1_2\otimes m,\lambda 1_2\oplus 1_2\otimes m]\otimes 1_N\label{reppif}
\end{equation}

where $[A,B,C,D]$ denotes a matrix which is block-diagonal with respect to the chirality space.   The real structure $J_F$ acts trivially on $\CC^2_{\mathrm{weak}}$ and $\CC^4_{\mathrm{colour}}$, and it acts as $(R,L,\bar R,\bar L)\mapsto(\bar R,\bar L,R,L)$ on the basis of $\CC^4_\chi$. Hence $J_F[A,B,C,D]J_F^{-1}=[C^*,D^*,A^*,B^*]$. Now if $u=(e^{i\theta},q,m)$ is a unitary element of $\mathcal{A}_{\mathrm{SM}}$, with $q\in SU(2)$ and $m\in U(3)$, we have

\begin{equation}
\Upsilon(u)=[\begin{pmatrix}
1& 0 \cr 0&e^{-2i\theta}
\end{pmatrix}\oplus \begin{pmatrix}
e^{i\theta}&0\cr 0&e^{-i\theta}
\end{pmatrix}\otimes m^*,qe^{-i\theta}\oplus q\otimes m^*,c.c.,c.c.]\otimes 1_N\label{UpsilonSM}
\end{equation} 

where c.c. means that the last two blocks are the complex conjugate of the first two. To get rid of the extra $U(1)$ we impose the \emph{unimodularity condition} $\det(\pi_F(u))=1$. This is equivalent (up to a finite abelian group, which we neglect here) to ask $m$ to be of the form $m=e^{-i\theta/3}g$, with $g\in SU(3)$, and  \eqref{UpsilonSM} becomes

\begin{equation}
\Upsilon(u)=[\begin{pmatrix}
1& 0 \cr 0&e^{-2i\theta}
\end{pmatrix}\oplus \begin{pmatrix}
e^{4i\theta/3}&0\cr 0&e^{-2i\theta/3}
\end{pmatrix}\otimes g^*,qe^{-i\theta}\oplus qe^{i\theta/3}\otimes g^*,c.c.,c.c.]\otimes 1_N \label{UpsilonSMunimod}
\end{equation}
where $g^*$ is a generic element of $SU(3)$. We see that this yields the correct representation of the gauge group, which is remarkable. Nonetheless, we would prefer to obtain the unimodularity condition in a natural way rather than by simply imposing it by hand. This would be possible if the finite algebra $\mathcal{A}_F$ was  complex and $\pi_F$ a complex representation. Indeed, we would have $\Upsilon(e^{i\theta }u)=e^{i\theta}\pi(u)Je^{i\theta}\pi(u)J^{-1}=\Upsilon(u)$, so that $\Upsilon$ would yield a representation of $U(\mathcal{A}_F)/U(1)$, which can be identified up to a finite center with $SU(\mathcal{A}_F)$. However, even ignoring the problem of the algebra $\HH$, we could not have $\CC$ act differently on up and down right-handed leptons in \eqref{reppif}, and thus we would not end up with the correct hypercharges. It thus seem that, as far as this problem is concerned, we have come to a dead end\footnote{In the case  of the Standard Model, the unimodularity condition is equivalent to the cancellation of quantum anomalies. We   could thus argue that the extra $U(1)$ is an artefact of the classical theory. Unfortunately this equivalence does not hold for all particle models.}.

\section{Algebraic backgrounds}\label{secAB}

Pre-spectral triples have been in use implicitly in noncommutative geometry since the first particle models were defined, but have been investigated explicitly only   recently \cite{AB1}. A number of problems have been found. To understand them, let us consider the case of  pure gravity on a manifold. Since we have a Hilbert space of spinors, we will consider tetradic gravity on a spin manifold. Pre-spectral triple automorphisms should correspond to the symmetries of this theory, i.e. to diffeomorphisms coupled with lifts of the spin group (sometimes called local Lorentz transformations in physics). However, they do not: as soon as $\dim(M)\ge 6$ there are extra local automorphisms \cite{AB1}. It means that we are  missing a background structure, left invariant by the true automorphisms. We can understand what it is by considering the algebraic definition of the spin group. Let $g$ be a metric (of possibly indefinite signature) on $\RR^n$, with $n$ even. Let us embed $\RR^n$ in the complex Clifford algebra $\CC l(\RR^n,g)$. The latter is equipped with a natural bilinear form $\tilde g$ extending $g$, a real structure $J$ and a chirality $\chi$ (see \cite{SSTpart1} for details). Moreover the spin group $\mathrm{Spin}(g)$ precisely contains the elements $u$ of $\CC l(\RR^n,g)$ such that :

\begin{enumerate}
\item $u$ is unitary,
\item $u$ commutes with $J$ and $\chi$,
\item $\mathrm{Ad}_u$ leaves $\RR^n$ invariant.
\end{enumerate}

If we translate these conditions in the case of the canonical spectral triple over a spin manifold, we obtain a pre-spectral triple automorphism with an additional property: it leaves the bimodule of $1$-forms $\Omega^1_D{\mathcal A}$ invariant. Despite the notation and equation \eqref{def1forms}, the latter does not really depend on $D$, but rather on the differential and spin structures\footnote{On a Riemannian manifold it is possible  to define spin structures independently of the metric \cite{LudwikSpin}. An even more straightforward solution is available (in any signature) if the manifold is parallelizable \cite{AB1}.}. This motivates us to define a new structure, falling in between a pre-spectral triple and a spectral triple, and which we call an \emph{algebraic background}. 

\begin{definition}
An \emph{algebraic background} (or background for short) is   a pre-spectral triple equipped with an $\mathcal{A}$-bimodule $\Omega^1$ which is odd\footnote{It must also have a non-trivial configuration space, see  definition \ref{defconfspace} below.} ({  i.e.} it anticommutes with $\chi$). A \emph{background automorphism} is a pre-spectral triple automorphism $U$ such that $U\Omega^1U^{-1}=\Omega^1$.
\end{definition}

With this definition,  automorphisms  of the canonical background over a manifold bijectively correspond to  products of diffeomorphisms with local Lorentz transformations, as expected. Moreover, in the case of the Standard Model background, we find in addition to the latter the gauge transformations as well as gauged B-L symmetries \cite{AB2}\cite[$\S 4.2$]{ShaneThesis}. There are {  many} more pre-spectral triple automorphisms in this case, notably flavour changing symmetries. Note that $\Upsilon(u)$ is   a background automorphism provided 

\begin{equation}
\pi(u)^o\Omega^1\pi(u^{-1})^o=\Omega^1\label{weakC1}
\end{equation}

which  is a weaker condition than $C_1$, so we call it \emph{weak $C_1$}.
An example of a model in which  $C_1$ fails  while \emph{weak $C_1$} continues to hold can be found in  \cite{weakC1}, in which  a $B-L$ extension of the standard model was obtained by enlarging the standard model coordinate algebra. We will discuss the relationship between $C_1$ and symmetries in more detail in section \ref{fluctuationsection}.

Let us summarize this section so far with the following analogy: just as GR is not defined on a Riemannian manifold, but on the configuration space of all metrics compatible with the differential structure of the background manifold, a generalized gravity model in NCG cannot be defined on a spectral triple, but instead on the configuration space of all Dirac operators compatible with a more primitive structure not including the Dirac operator. There are two options (as of yet): pre-spectral triples and algebraic backgrounds. We prefer the latter since it gives better results as far as automorphisms are concerned. Hence we will view algebraic backgrounds as the noncommutative equivalent of the background manifold with its differential (and spin) structure. It should be noted that no reconstruction theorem in the spirit of \cite{Connes:reconstruction} has been proven to date for this structure\footnote{However, in the commutative \emph{and finite} case, an algebraic background encodes a finite graph \cite{BesnardCorfu} while a pre-spectral triple only encodes a set of points, and it is known that giving a differential structure to a finite set is equivalent to turning it into a graph \cite{Dimakis}.}.

\section{Configuration space and fluctuations}\label{secConfigFluct}

Algebraic backgrounds, or, for that matter, pre-spectral triples, permit  one to define a configuration space for a  generalized gravity theory in a natural way.

\begin{definition}\label{defconfspace} The \emph{configuration space} of an algebraic background is the space of all operators $D$ satisfying the conditions in the definition of a spectral triple and such that $\Omega^1_D{\mathcal A}\subset \Omega^1$. If $D$ is such that $\Omega^1_D{\mathcal A}=\Omega^1$, we say that $D$ is \emph{regular}. If there exists a regular $D$ on a background, we say that the background is regular.
\end{definition}

The configuration space of a pre-spectral triple can be similarly defined, except that the condition $\Omega^1_D{\mathcal A}\subset \Omega^1$ is replaced with $C_1$. It is worthy of note that the two definitions give the same result in the case of a manifold (compare \cite{costarica}, chapter 11 and \cite{AB1}). 

The configuration space of a background or pre-sprectral triple is immediately checked to be invariant under the relevant automorphisms. In particular, if $D$ is a Dirac operator and $u$ is a unitary element of the algebra, then $\Upsilon(u)D\Upsilon(u)^{-1}$ is a Dirac operator. Let us write it in the form $D+F_{u}$ where

\begin{equation}
F_u:=\Upsilon(u)D\Upsilon(u)^{-1}-D\label{AssocInnerFluct1}
\end{equation}
 is called an \emph{inner fluctuation} of $D$, since it is associated to an inner automorphism of the algebra, while a Dirac operator like $D+F_u$ is called a \emph{fluctuated Dirac}. It is from such fluctuations that non-gravitational bosonic fields arise in noncommutative geometry.  Physically, fluctuations like \eqref{AssocInnerFluct1} only yield pure gauge fields, so we need a larger space. The standard choice is $D+\mathcal{F}_D$, where 

\begin{equation}
\mathcal{F}_D=\{\omega+\omega^o|\omega\in\Omega^1, \omega^\dagger=\omega\}.\label{DefFluctAssoc}
\end{equation}
An element of $\mathcal{F}_D$ is a called a (general)  \emph{fluctuation}. If we neglect gravity, it is enough to define the physical theory on a   subspace of the configuration space containing all fluctuated Dirac operators $D+F$, where $F\in \mathcal{F}_D$, and with a fixed $D$ (corresponding for instance to the Minkowski metric). We can justify the choice given in Eq.~\eqref{DefFluctAssoc}, by observing that fluctuations given in \eqref{AssocInnerFluct1}  are indeed of the form $\omega+\omega^o$, with $\omega=\pi(u)[D,\pi(u)^{-1}]$, and that $D+\mathcal{F}_D$ is an affine space invariant under the relevant  symmetries, namely the background automorphisms fixing $D$. Also observe that if $u=\exp(a)$, with $a^\dagger=-a$, then $\mathcal{F}_D$ contains $[D,\pi(a)]+[D,\pi(a)]^o$ which is the infinitesimal version of \eqref{AssocInnerFluct1}. Note that $\mathrm{ad}_a$ is an inner derivation of $\mathcal{A}$ which commutes with the $*$-operation, and that such derivations exponentiate to $*$-automorphisms. There is thus the following commutative diagram:

\begin{equation}
\xymatrix{{\mathrm Aut}_{\mathrm Inn}( \mathcal{A})& U(\mathcal{A})\ar[l]_{\mathrm{Ad}}\ar[r]^{\Upsilon} &  {\mathcal U}(\mathcal{A}) \cr
	{\mathrm Der}_{\mathrm Inn}(\mathcal{A})\ar[u]^{\exp}&\ar[l]_{\mathrm{ad}} \mathrm{Skew}(\mathcal{A})\ar[u]^{\exp} \ar[r]^{d_e\Upsilon}& {\mathcal G}(\mathcal{A})\ar[u]^{\exp}
}\label{assocliftdiag}
\end{equation}

where have used the following notations: $\mathrm{Skew}(\mathcal{A})$ is the space of anti-selfadjoint elements of $\mathcal{A}$, $\mathcal{U}(\mathcal{A})$ is $\Upsilon(U(\mathcal{A}))$ (and is a subgroup of $\mathrm{Aut}(S)$ which we can call the subgroup of gauge transformations), $d_e\Upsilon(a)=a-a^o$ is the differential of $\Upsilon$ at the identity, and  $\mathcal{G}(\mathcal{A}):=d_e\Upsilon(\mathrm{Skew}(\mathcal{A}))$ is the space of infinitesimal gauge transformations. Note that both $\mathcal{U}(\mathcal{A})$ and $\mathcal{G}(\mathcal{A})$ act on $\mathcal{F}_D$ thanks to the bimodule structure of $\Omega^1_D$. However, since  the arrows on the left cannot be inversed (because of central elements), we cannot make the ${\mathrm Aut}_{\mathrm Inn}( \mathcal{A})$ or ${\mathrm Der}_{\mathrm Inn}(\mathcal{A})$ directly act on fluctuations.

\section{Preliminaries on Jordan algebras}\label{prelimJB}
\subsection{Generalities}
 A (real) Jordan algebra\footnote{In this paper we will only consider real Jordan algebras.} is a real vector space $A$ equipped with a bilinear, commutative product $\circ$ satisfying the Jordan identity \cite{AlfsenShultz}:

\begin{equation}
\forall a,b\in A,\ (a^2\circ b)\circ a=a^2\circ(b\circ a)\label{JordanIdentity}
\end{equation}

Note that while Jordan algebras are in general not associative,  every Jordan algebra is power associative, i.e. $a^n$ has an unambiguous meaning for all $n\in\mathbb{N}$. All the Jordan algebras considered in this paper will be unital.

Let $\mathcal{A}$ be an associative algebra equipped with the product 
\begin{equation}
a\circ b=\frac{1}{2}(ab+ba)\label{circprod}
\end{equation}
and let $A$ be a real subspace of $\mathcal{A}$ stable under $\circ$ (an important example is when $\mathcal{A}$ is a $*$-algebra and $A$ is the space of selfadjoint elements). Then $(A,\circ)$ is a Jordan algebra. A Jordan algebra isomorphic to one of this kind is called \emph{special}, otherwise it is called  \emph{exceptional}. We will be primarily interested in the special case in this paper.

Let us now introduce some notations and examples. For $\mathbb{K}=\RR,\CC,\HH$ and $n\in \NN^*$, we denote by $H_n(\mathbb{K})$ the Jordan algebra of selfadjoint elements of $M_n(\mathbb{K})$. We similarly define $H_3(\mathbb{O})$ where $\mathbb{O}$ is the non-associative algebra of octonions. It turns out that $H_3(\mathbb{O})$ is a Jordan algebra (called the Albert algebra), which is exceptional.   Let us also introduce the \emph{spin factor} $\mathrm{JSpin}(n)$, which is the sub-Jordan algebra generated by $\RR^n$ in the Clifford algebra $\mathrm{Cl}(\RR^n,g)$, where $g$ is the canonical scalar product, equipped with \eqref{circprod}. It is thus generated by vectors of $\RR^n$ with the product $u\circ v=g(u,v)$. 

Let us now introduce \emph{Jordan-Banach} algebras (or JB-algebras).

\begin{definition} A JB algebra is a normed Jordan algebra $A$ which is complete in the norm and satisfies, for all $a,b\in A$,
\begin{enumerate}
\item $\|a\circ b\|\le \|a\|\|b\|$,
\item $\|a^2\|=\|a\|^2$,
\item $\|a^2\|\le \|a^2+b^2\|$.
\end{enumerate}
\end{definition}

The main example of a JB alebra is the selfadjoint part of a $C^*$-algebra equipped with the symmetrized product given in Eq.~\eqref{circprod}. Just as with $C^*$-algebras, JB algebras admit a continuous functional calculus. Moreover, just as commutative $C^*$-algebras correspond to algebras of complex functions on compact Hausdorff spaces by the Gelfand transform, the same is true of associative JB algebras. More precisely,   a JB algebra is associative iff it is isomorphic to the algebra of real functions on a compact Hausdorff space (see~\cite{Farnsworth:2020ozj} and references therein for a more in depth discussion of the relationship between JB algebras and $C^*$-algebras). Finite-dimensional JB algebras can be completely classified.

\begin{theorem}\label{ClassificationTheorem}
Every finite-dimensional JB-algebra is a direct sum of ones on this list:
\begin{enumerate}
\item $H_n(\mathbb{K})$ for $\mathbb{K}=\RR,\CC$ or $\HH$,
\item $\mathrm{JSpin}(n)$,
\item $H_3(\OO)$.
\end{enumerate}
\end{theorem}
\begin{rem} Note that the same theorem holds for finite-dimensional formally real Jordan algebra, i.e. Jordan algebras such that $a_1^2+\ldots+a_k^2=0\Rightarrow a_1=\ldots=a_k=0$ for a any sum of squares.
\end{rem}

The following definitions can be found\footnote{Except for  a reversal in the products  due to a  different  convention   for map composition.} in \cite{jacobson1968structure}, chap. 2.

\begin{definition} Let $A$ be a Jordan algebra and $\mathcal{A}$ be an associative algebra. A linear map $\sigma : A\rightarrow \mathcal{A}$ is called
\begin{enumerate}
\item an \emph{associative specialization} of $A$ in $\mathcal{A}$ if for all $a,b\in A$,

\begin{equation}
\sigma(a\circ b)=\sigma(a)\circ \sigma(b)\label{assocrep}
\end{equation}
\item a \emph{multiplicative specialization} of $A$ in $\mathcal{A}$ if for all $a,b\in A$,

\begin{eqnarray}
[\sigma(a),\sigma(a^2)]&=&0,\cr
2\sigma(a)\sigma(b)\sigma(a)+\sigma(a^2\circ b)&=&2\sigma(a\circ b)\sigma(a)+\sigma(a^2)\sigma(b)\label{multrep}
\end{eqnarray}
\end{enumerate}
\end{definition}

\begin{theorem}\label{JacTh14}
Let $\sigma_1$ and $\sigma_2$ be two associative specializations of $A$ in $\mathcal{A}$ such that $[\sigma_1(A),\sigma_2(A)]=0$. Then $\rho=\frac{1}{2}(\sigma_1+\sigma_2)$ is a multiplicative specialization of $A$ in $\mathcal{A}$.
\end{theorem}

An immediate corollary is that $\frac{1}{2}$ times an associative specialization is automatically a multiplicative specialization, so that the latter are more or less a generalization of the former. A  specialization in  ${\rm End}(H)$, will be called a \emph{representation}. The reason for such a terminology is the relationship between multiplicative representations and Jordan modules. There is a general theory for modules over non-associative algebras \cite{eilenberg}. In the Jordan case this boils down to the following definition.

\begin{definition}\label{defjordanmodule} Let $A$ be a Jordan algebra, $H$ be a vector space and $S : A\rightarrow {\rm End}(H)$ be a linear map. Let $\circ$ be the  bilinear product on $A\oplus H$ extending $\circ$ on $A$ and such that $a\circ h=h\circ a=S_ah$, $h\circ h'=0$, for all $a\in A$, $h,h'\in H$. We say that $(H,S)$ is a \emph{Jordan module}  if $(A\oplus H,\circ)$ is a Jordan algebra.
\end{definition}

\begin{theorem} With the notations of definition \ref{defjordanmodule}, $(H,S)$ is a Jordan module iff $S$ is a multiplicative representation.
\end{theorem}

An equivalent requirement is \cite{carotenuto2019}:
\begin{eqnarray}
[S_a,S_{b\circ c}]+[S_c,S_{a\circ b}]
+[S_b,S_{c\circ a}]&=&0,\cr
S_aS_bS_c+S_cS_bS_a + S_{(ac)b}&=&S_aS_{bc}+S_{b}S_{ac}+S_{c}S_{ab},\label{equivmultrep}
\end{eqnarray}

%\begin{eqnarray}
%S_{x^3}-3S_{x^2}S_x+2S_x^3&=&0\cr
%[[S_x,S_y],S_z]+S_{[x,z,y]}&=&0\label{equivmultrep}
%\end{eqnarray}

which is   obtained by linearization of \eqref{multrep}. Summing all the cyclic permutations of the second equation in \eqref{equivmultrep} and using Jacobi's identity also entails 
\begin{align}
[[S_x,S_y],S_z]=S_{[y,z,x]},\label{order0}
\end{align}
where \emph{the associator} is defined to be $[x,y,z]:=(xy)z-x(yz)$.  We will make extensive use of \eqref{order0} in what follows.  Note also the following natural properties.
 
\begin{proposition}\label{multamod} Let ${A}$ be a Jordan algebra and ${ A'}$ be a subalgebra. Then ${ A}$ is an ${A'}$-module for the action $L_ab=a\circ b$.
\end{proposition}

A Jordan algebra is thus a Jordan module over itself. 

\begin{proposition} Let $A,{A'}$ be   Jordan algebras and $(H,S)$ be a ${A'}$-module. Let $\phi : A\rightarrow {A'}$ be a homomorphism. Then $(H,S\circ \phi)$ is an $A$-module.
\end{proposition}

It follows from this that if $\sigma$ is an associative representation of $A$ on $H$, then $B(H)$ is an $A$-module for the action $S_aT=\sigma(a)\circ T$, with $T\in B(H)$. The reason is that $(B(H),\circ)$ is a Jordan algebra, hence a Jordan module over itself.

 Let $(H,S)$ and $(H',S')$ be two modules over $A$ and $f : H\rightarrow H'$ be a linear map. Then $f$ is defined to be a \emph{module homomorphism}, or module map, iff $f(S_a h)=S_a'f(h)$ for all $(a,h)\in A\times H$. This is equivalent to requiring $\mathrm{Id}\oplus f$ to be a   homomorphism from $A\oplus H$ to $A\oplus H'$ seen as Jordan algebras. The image and kernel of module maps are submodules, the quotient of a module by a submodule is a module: all of these results extend without change from the associative to the Jordan case. However care must be taken with the tensor product {in the non-associative setting}. For instance, let $M$ be a vector space. We want to define the \emph{free $A$-module generated by $M$} to be an $A$-module $\bra M\ket_A$ with a linear inclusion map $\iota : M\rightarrow \bra M\ket_A$ satisfying the following universal property: for all $A$-modules $H$ and linear maps $f : M\rightarrow H$, there exists a module map $\tilde f : \bra M\ket_A\rightarrow H$ extending $f$, i.e. such that the following diagram commutes:

\begin{equation}
\xymatrix{M\ar[d]^\iota\ar[dr]^f & \cr
\bra M\ket_A\ar[r]^{\tilde f}&H}\label{UPfreemod}
\end{equation}
 
For an associative algebra $\mathcal{A}$, such an object exists and can be taken to be $\mathcal{A}\otimes_\RR M$, with inclusion map $m\mapsto 1\otimes m$. The extension of $f$ is defined by $\tilde f(a\otimes m)=af(m)$. However, when $\mathcal{A}$ is not associative this fails to be a module map. We do not know if there exists an object with the universal property \eqref{UPfreemod} in all generality in the Jordan category, however we provide below a solution in the special case, which will be useful in section \ref{secJBT} when defining universal $1$-forms. Let $A$ be special Jordan algebra embedded in an associative algebra $\mathcal{A}$. If there is no natural choice for  $\mathcal{A}$ in the context at hand, one can always consider the universal associative enveloppe of $A$ \cite{jacobson1968structure}.   A \emph{special Jordan module} for $A$ will  be an $A$-module $(H,S)$ which is also an $\mathcal{A}$-bimodule with $S_a h=\frac{1}{2}(a\cdot h+h\cdot a)$ for all $(a,h)\in A\times H$ (see also  \cite{jacobson1968structure}, p 100). Now for any vector space $M$, let $\bra M\ket_{\mathcal A}:=\mathcal{A}\otimes \mathcal{A}^{opp}\otimes M$, where the superscript $opp$ denotes the opposite algebra. We have the inclusion map $m\mapsto 1\otimes 1\otimes m$, and it is easy to see that  $\bra M\ket_{\mathcal A}$ is canonically an $\mathcal{A}$-bimodule with the following universal property: for any linear map $f : M\rightarrow H$, with $H$ an $\mathcal{A}$-bimodule, there exists a bimodule map $\tilde f$ such that 
\begin{equation}
\xymatrix{M\ar[d]^\iota\ar[dr]^f & \cr
\bra M\ket_{\mathcal A}\ar[r]^{\tilde f}&H}\label{UPfreebimod}
\end{equation}
commute. Being an $\mathcal{A}$-bimodule, $\bra M\ket_{\mathcal A}$ is also an $A$-module containing $M$. We define $\bra M\ket_{A}$ to be   the sub-$A$-module of $\bra M\ket_{\mathcal A}$ generated by $M$. We call it the \emph{special free $A$-module generated by $M$}. Since $\tilde f$ satisfies, for all $a\in A$ and $m\in M$:
\begin{eqnarray}
\tilde f(S_a m)&=&\frac{1}{2}\tilde f(a\cdot m+m\cdot a)\cr
&=&\frac{1}{2}(\tilde{f}(a)m+m\tilde{f}(a))\cr
&=&S_a\tilde{f}(m)
\end{eqnarray}
it is an $A$-module map. Hence $\bra M\ket_{A}$ satisfies the universal property \eqref{UPfreemod} for all special Jordan module $H$ and linear map $f$.  Its universal property  then makes it unique up to isomorphism, as is usual.

\subsection{Derivations and automorphisms of JB algebras}

We start by recalling some facts about order-derivations. The following definition was introduced in \cite{Connes:1974} (see also \cite{AlfsenShultz}).

\begin{definition} A bounded linear operator $\delta$ on a JB-algebra $A$ is called an \emph{order derivation} if for all $t\in\RR$, $e^{t\delta}(A^+)\subset A^+$, where $A^+$ is the subset of positive elements of $A$.
\end{definition}

Note that in that case, $e^{t\delta}$ will be an \emph{order automorphism}, i.e. a bijective linear map $\phi$ such that $\phi$ and $\phi^{-1}$ preserve  order.  From proposition \ref{multamod}, every Jordan algebra acts as a module over itself. For every $x$ in a Jordan algebra $A$,  we write $L_x$  for the Jordan multiplication by $x$, i.e. $L_xy = x\circ y$ for $x,y\in A$. In this case, $L_x$ is an order-derivation.

\begin{definition} An order-derivation on a unital JB-algebra $A$ is called self-adjoint iff it is of the form $L_a$, for some $a\in A$, and it is called \emph{skew}, if $\delta(1)=0$.
\end{definition}

\begin{proposition} The set $\mathrm{Der}(A)$ of order-derivations of a JB-algebra $A$ is norm closed and  closed under the  Lie bracket. Moreover, it is the direct sum $\mathrm{Der}(A)_{sa}\oplus \mathrm{Der}(A)_{skew}$ of the real subpaces of selfadjoint and skew order-derivations. 
\end{proposition}
\begin{proof} See \cite{AlfsenShultz} prop. 1.59 and 1.60.
\end{proof}

Using the fact that unital order automorphisms are Jordan automorphisms (see \cite{AlfsenShultz}, Th. 2.80), it is possible to give the following characterization of skew derivations:

\begin{proposition} Let $A$ be JB-algebra and $\delta$ an order derivation on $A$. The following are equivalent\label{propaut}:
\begin{enumerate}
\item $\delta$ is skew,
\item $\forall t\in\RR$, $e^{t\delta}$ is a Jordan automorphism,
\item $\delta$ is a Jordan derivation,
\end{enumerate}
where a Jordan automorphism on a Jordan algebra $A$ is an invertible linear map $\alpha:A\rightarrow A$, which respects the product on $A$:
\begin{align}
\alpha(ab) = \alpha(a)\alpha(b),
\end{align}
and a Jordan  derivation is a linear map $\delta:A\rightarrow A$, which satisfies the Leibniz rule:
\begin{align}
\delta(ab) = \delta(a)b + a\delta(b).
\end{align}

\end{proposition}

\begin{definition} An \emph{inner derivation} of $A$ is a Jordan derivation of the form
$$\delta=\sum_i [L_{a_i},L_{b_i}].$$
The Lie algebra of inner derivation will be denoted by $\mathrm{Der}_{Inn}(A)$.\label{innerdeff}
\end{definition}

The above definition for inner derivations is standard, and generalizes to all Jordan algebras.  Similarly, we have the following definition.

\begin{definition} The   group ${\mathrm Aut}_{\mathrm Inn}(A)$  of inner automorphisms of $A$ is by definition the subgroup of order automorphisms of $A$ generated by exponentionals of inner derivations. \label{iinnerauts}
\end{definition}

The next definition can be found in \cite{Schafer}.

\begin{definition} The \emph{Lie multiplication algebra} $M(A)$ is the sub-Lie algebra of $\mathrm{End}(A)$ generated by the operators $L_a$, $a\in A$. 
\end{definition}

The Lie multiplication is equal to the direct sum $\mathrm{Der}(A)_{sa}\oplus \mathrm{Der}_{Inn}(A)$.  If $A$ is  finite-dimensional (or more generally if it is a JBW algebra, see \cite{AlfsenShultz} for the   definition), then $M(A)=\mathrm{Der}(A)$.

We will now apply the previous concepts to the case where $A$ has an  associative representation in a Hilbert space.

\begin{definition}
Let  $\pi : A\rightarrow B(H)$ be an associative representation of the Jordan algebra $A$. Then we define $\mathrm{Lie}_\pi(A)$ to be the Lie algebra generated by $\pi(A)$ in $B(H)$.
\end{definition}

To prove the next proposition we need to recall that for $x,y\in\pi(A)$, 

\begin{equation}
[L_x,L_y]=\frac{1}{4}ad_{[x,y]}\label{LxLy}
\end{equation}

\begin{proposition}\label{propertiesofLiePi} $\mathrm{Lie}_\pi(A)=\pi(A)\oplus[\pi(A),\pi(A)]$. Moreover  $[\pi(A),[\pi(A),\pi(A)]]\subset \pi(A)$ and $[\pi(A),\pi(A)]$ is a Lie subalgebra.
\end{proposition}
\begin{proof}
Let $x,y,z,t\in \pi(A)$. Then $[x,[y,z]]=-4[S_y,S_z]x$ by \eqref{LxLy}. Hence $[x,[y,z]]=-4y\circ(z\circ x)+4z\circ (y\circ x)\in \pi(A)$ since $\pi(A)$ is a Jordan algebra. Moreover $[[x,y],[z,t]]=[[[x,y],z],t]+[z,[[x,y],t]]$ by Jacobi's identity. Now $[[x,y],z]$ and $[[x,y],t]$ are both in $\pi(A)$ by the above, hence  $[[x,y],[z,t]]\in[\pi(A),\pi(A)]$. This proves the result by linearity.
\end{proof}

For all $a,g,x\in B(H)$   let us define $\starad_a(x)$ and $\starAd_g(x)$ respectively by

\begin{equation}
\starad_a(x)=ax+xa^\dagger,\ \starAd_g(x)=gxg^\dagger
\end{equation}

Note that $\starAd_g$ is the natural action of $g$ on $B(H)$ if elements of $B(H)$ are interpreted as sesquilinear forms through $x\mapsto \langle .,x.\rangle$, rather than operators on $H$. 

\begin{proposition} $\starad_a$ defines a Lie homomorphism from $\mathrm{Lie}_\pi(A)$ onto $M(\pi(A))\subset\mathrm{Der}(\pi(A))$. 
\end{proposition}
\begin{proof}
Note that if $a=a^\dagger$, $\starad_a=2L_a$ and if $a=-a^\dagger$, $\starad_a=\ad_a$, so that $\starad_a$ is always an order derivation of $\pi(A)$. A direct computation shows that  $\starad_{[a,b]}x=[a,b]x+x[a,b]^\dagger=[\starad_a,\starad_b]x$.
\end{proof}

Note that by definition,  the image of $\pi(A)$ is   $\mathrm{Der}(\pi(A))_{sa}$ and the image of $[\pi(A),\pi(A)]$ is   $\mathrm{Der}_{inn}(\pi(A))$. In the cases of interest, $\pi$ will be faithful so that $\starad$ can be thought of as a surjective morphism from $\mathrm{L}ie_\pi(A)$ to $M(A)$, sending an operator on $H$ to an operator on $A$. In the case that interests us the most, it is 1-1.

\begin{proposition}\label{invertibility} Let $A$ be finite-dimensional or of the form $\mathcal{C}(M,A_F)$ where $A_F$ is finite-dimensional and $M$ is a Hausdorff space. Then $\starad : \mathrm{L}ie_\pi(A)\rightarrow M(\pi(A))$ is an isomorphism. In particular, $\ad :  [\pi(A),\pi(A)]\rightarrow \mathrm{Der}_\mathrm{Inn}(\pi(A))$ is an isomorphism, whose inverse will be denoted by $a$.
\end{proposition}
\begin{proof}
First let us observe than an element $z\in \ker\starad$ is skew, since $\starad_z(1)=z+z^\dagger=0$. Now if $A$ is finite-dimensional, and $z=\sum_i[x_i,y_i]\in\ker(\mathrm{ad})$, then $z$ will commute with the $C^*$-algebra generated by $\pi(A)$ and hence $\theta(z)$ will be a multiple of the identity for every irreducible representation $\theta$ of $A$. Since $\mathrm{Tr}(\theta(z))=0$, we have $\theta(z)=0$ for all $\theta$ irreducible, hence $z=0$.

In the second case, if $g\in\ker\starad$, then $g$ is skew for the same reason as above. Hence $g\in\mathcal{C}(M,[\pi(A_F),\pi(A_F)])$ such that for every $f\in A$ and every $x\in M$, $[g(x),f(x)]=0$. In particular if $f$ is constant, we see that $g(x)$ commutes with $\pi(A_F)$, so that the above argument yields $g(x)=0$.
\end{proof}

\begin{rem}
In the general infinite-dimensional case we can say the following : if $[x,y]\in\ker\mathrm{ad}$ then $[[x,y],x]=0$ which implies that $[x,y]$ is quasi-nilpotent (Kleinecke-Shirokov theorem), hence $[x,y]=0$ since $i[x,y]$ is selfadjoint. We do not know however if this result extends to sums of commutators.
\end{rem}

%\begin{proposition} The map $x\mapsto 2L_x$ extends as a Lie-homomorphism $\Phi$ from $\mathrm{Lie}_\pi(A)$ to $\mathrm{Der}(A)$.
%\end{proposition}
%\begin{proof}
%The map $\Phi : x\mapsto 2L_x$ on $\pi(A)$ and $T\mapsto  ad_T$ on $[\pi(A),\pi(A)]$ is certainly well-defined and \eqref{LxLy} shows that $[\Phi(x),\Phi(y)]=\Phi([x,y])$ for all $x,y\in\pi(A)$. Now let us prove that $\Phi([x,[y,z]])=[\Phi(x),\Phi([y,z])]$. The following formula will be useful:
%
%
%\begin{equation}
%[L_x,[L_y,L_z]]=L_{[y,x,z]}\label{Lassociator}
%\end{equation}
%
%where $[y,x,z]$ is the \emph{associator}
%
%\begin{equation}
%[y,x,z]=(y\circ x)\circ z-y\circ(x\circ z)
%\end{equation}
%
%We already know by proposition \ref{propertiesofLiePi} that $[x,[y,z]]\in \pi(A)$. Hence we have
%
%\begin{eqnarray}
%\Phi([x,[y,z]])&=&2L_{[x,[y,z]]},\text{ by definition of }\Phi,\cr
%&=&-8L_{[L_y,L_z]x},\text{ by }\eqref{LxLy},\cr
%&=&8L_{[y,x,z]},\cr
%&=&[2L_x,4[L_y,L_z]],\cr
%&=&[\Phi(x),\Phi([y,z])]
%\end{eqnarray}
%
%Finally we have
%
%\begin{eqnarray}
%\Phi([[x,y],[z,t]])&=&ad_{[[x,y],[z,t]]}\cr
%&=&[ad_{[x,y]},ad_{[z,t]}],\text{ by Jacobi's identity},\cr
%&=&[\Phi([x,y]),\Phi([z,t])]
%\end{eqnarray}
%
%It follows that $\Phi$ is a Lie homomorphism.
%
%\end{proof}

Now let us define the group  $\mathcal{L}\mathrm{ie}_\pi(A)$, which is generated by $\exp(\mathrm{Lie}_\pi(A))$, and its subgroup  $U(A)$ which is also a subgroup of $U(H)$ and is  generated by $\exp([\pi(A),\pi(A)])$. Then the following diagram commutes:

\begin{equation}
\xymatrix{{\mathrm Aut}_{order}(\pi(A)) & \mathcal{L}\mathrm{ie}_\pi(A) \ar[l]_{\starAd}  \cr
M(\pi(A))\ar[u]^{\exp} & \mathrm{Lie}_\pi(A) \ar[l]_{\starad}\ar[u]^{\exp}
}\label{diagramstarAd}
\end{equation}

Over the subalgebra $[\pi(A),\pi(A)]$, \eqref{diagramstarAd} reduces to:

\begin{equation}
\xymatrix{{\mathrm Aut}_{\mathrm Inn}(\pi(A)) & U(A)\ar[l]_{\mathrm{Ad}}  \cr
{\mathrm Der}_{\mathrm Inn}(\pi(A))\ar[u]^{\exp} & [\pi(A),\pi(A)]\ar[l]_{\mathrm{ad}}\ar[u]^{\exp}
}\label{diagraminnder}
\end{equation}

When $\mathrm{ad}$ is $1-1$, since $\exp$ also is near the identity, we can infer that $\mathrm{Ad}$ will be invertible near $1$. In this case we obtain the following commutative diagram, where the wriggling arrow is defined near the identity.
 
\begin{equation}\label{diagraminversion}
\xymatrix{{\mathrm Aut}_{\mathrm Inn}(\pi(A))\ar@{~>}[r] & U(A)  \cr
{\mathrm Der}_{\mathrm Inn}(\pi(A))\ar[u]^{\exp}\ar[r]^{a} & [\pi(A),\pi(A)]\ar[u]^{\exp}}
\end{equation}

\section{Jordan background and   triples}\label{secJBT}

After these preliminaries we come to the main course. We begin in this section with special Jordan triples, {  leaving the discussion of the generalization to the exceptional case for section \ref{secbeyond}}.

\begin{definition} A (real, even) \emph{special Jordan triple} is gadget $\mathcal{T}=(A,H,\pi,D,\chi,J)$ such that
\begin{enumerate}
\item $A$ is a special Jordan algebra,
\item $H$ is a Hilbert space,
\item $\pi$ is a  faithful associative representation  of $A$,
\item $\overline{\pi(A)}$ is a JB algebra,
\item $D,J$ and $\chi$ are the usual things, with the usual conditions (in particular C0, but not C1 yet),
\item For all $a\in A$, $ [D,\pi(a)]$ is well-defined and bounded.
\end{enumerate}
\end{definition}

\begin{rem}  We need to be careful  here as $A$ (whose elements plays the role of differentiable functions)  is not a JB-algebra as in section \ref{prelimJB}. However, it can readily be checked that the main result of this section, namely proposition \ref{invertibility}, still holds with an   algebra of the form $\mathcal{C}^\infty(M,A_F)$, with an unchanged proof.
\end{rem}

We now turn to $1$-forms. There is a standard way to 
construct a graded algebra of differential forms $\Omega_d A$ from an algebra $A$, which continues to make sense in the Jordan setting~\cite{carotenuto2019}.   In this paper we will not consider forms of degree higher than one.  We obviously have $\Omega^0_dA = A$. We want the space $\Omega^1_dA$ of one-forms to be generated, as a Jordan $A$-module, by symbols  $d[a]$, $a \in A$, with relations
\begin{align}
d[a\circ b] &= d[a]\circ b + a\circ d[b] ,&  \forall a, b& \in A,\label{oneform1}\\
d[\alpha a + \beta b] &= \alpha d[a] + \beta d[b] ,& \forall a, b &\in A ,~ \alpha,\beta\in\mathbb{R},
\end{align}
where the product `$\circ$' is symmetric and  satisfies the Jordan identity. That is, we want $d : \Omega^0_d A\rightarrow \Omega^1_dA$ to be a derivation.  A generic element $\omega\in \Omega^1_d A$ is a finite sum of the form
\begin{align}
\omega = \sum a_1\circ (a_2\circ (...(a_{n-1}\circ d[a_n])))\label{generic1form}
\end{align}
with $a_i\in \Omega^0_d A$. Here is an explicit construction  of $\Omega^1_dA$. We first consider the vector space $dA:=A/\RR.  1_A $ and call $d$ the quotient map. We now make use of the fact that   $A$ is a special Jordan algebra embedded in $\mathrm{End}(H)$  and consider the special free module $\bra dA\ket_A$ defined in section \ref{prelimJB}. We let $R$ be the sub-module of  $\bra dA\ket_A$ generated by $d(a\circ b)-d(a)\circ b - a\circ d(b)$ for all $a,b\in A$, and finally we set $\Omega^1_dA:=\bra dA\ket_A/R$. By construction it is an $A$-module and $d$ has been extended to a derivation of $A$ into it. Moreover, it satisfies the following universal property.

\begin{theorem} Let $\mathcal{M}$ be any special $A$-module and $\delta : A\rightarrow \mathcal{M}$ be a derivation. Then there exists an $A$-module map $p : \Omega^1_dA\rightarrow \mathcal{M}$ such that the following diagram commute:
\begin{equation}
\xymatrix{A\ar[d]^d\ar[dr]^\delta & \cr
\Omega^1_dA\ar[r]^{p}&\mathcal{M}}\label{UPomega1}
\end{equation}

\end{theorem} 
\begin{proof}
For all $a\in A$ we set $p(da):=\delta(a)$. This is well-defined because $d$ annihilates constants. Since $\delta$ and $d$ are linear, $p$ also is, and we thus have a well-defined linear map $p : dA\rightarrow \mathcal{M}$. By the universal property \eqref{UPfreemod} of $\bra dA\ket_A$, we obtain a module map which we still call $p$ from $\bra dA\ket_A$ to $\mathcal{M}$. Since $\delta$ is a derivation, $R$ is in the kernel of $p$ which thus goes to the quotient, ending the proof.
\end{proof}

For this reason, elements of $\Omega^1_dA$ will be called \emph{(special) universal $1$-forms.} Note that a module of universal forms in the above sense has been constructed in \cite{carotenuto2019} for the Albert algebra (using a different approach), so that this work is complementary to ours in this respect. Now, we can make use of the previous theorem to define two natural representations of $1$-forms on $H$. First, we observe that since $\pi$ is an associative representation, the map
\begin{equation}
\delta(a):=[D,\pi(a)]\label{someq}
\end{equation}
is a derivation of $A$ into $\mathrm{End}(H)$, so that by the universal property of $\Omega^1_d A$ there exists a well-defined module map $p$ such that

\begin{align}
	p(da) &= [D,\pi(a)],& a&\in A.\label{da}
\end{align}
We will write $\pi(da)$ instead of $p(da)$ since no confusion should arise. We let $\Omega_D^1 A$ be the image of this map $\pi$. In  other words, $\Omega_D^1 A$ is given by the Jordan  $A$-submodule of $\mathrm{End}(H)$ generated by $\{[D,\pi(a)], a\in A\}$. Following the definition of an associative representation given in Eq.~\eqref{assocrep}, a general element in $\Omega_D^1 A$ is then  given by
\begin{align}
\pi( a_1\circ (a_2\circ (...(a_{n-1}\circ d[a_n])))) =  \sum \pi(a_1)\circ (\pi(a_2)\circ (...(\pi(a_{n-1})\circ[D,\pi(a_n)])))
\end{align}  

For the case in which $\Omega^0 A = A$ is a special Jordan algebra, we will continue to refer to  the corresponding differential, graded, Jordan algebras $\Omega = \Omega^0 A\oplus \Omega^1 A \oplus...$ as `special'. 
  Notice that the above definition collapses to the standard and familiar form when dealing with associative, Jordan algebras (e.g. in the canonical setting of a Riemannian manifold). Note, furthermore, that because  $\Omega_D^1 A$ acts as a Jordan module over $A$, it will also act naturally as a Lie module over $Der_{Inn} (A)$ for  the exact same reasons that $H$ acts as a Lie module over  $Der_{Inn} (A)$.

{Furthermore}, there is another perfectly natural representation of universal $1$-forms on $H$ based on changing the embedding of $A$ in $\mathrm{End}(H)$ from $\pi$ to $\pi^o$. This will yield a module map $p$ such that $p(da)= [D,\pi^o(a)]$. We will also write it $\pi^o$ instead of $p$, but this time we need to be careful since $\pi^o(da)=[D^o,\pi(a)^o]=-[D,\pi(a)]^o=-\pi(da)^o$, so that $\pi^o$  is not the composition of $\pi$ with $o$ on $1$-forms.

\begin{definition} A (real, even)  \emph{special Jordan background} is a gadget $\mathcal{B}=(A,H,\pi,\Omega^1,\chi,J)$ such that
\begin{enumerate}
\item $A$ is a special Jordan algebra,
\item $H$ is a Hilbert space,
\item $\pi$ is a {faithful} associative representation  of $A$,
\item $\overline{\pi(A)}$ is a JB algebra,
\item $J$ and $\chi$ are the usual things, with the usual conditions (in particular $C_0$),
\item $\Omega^1$ is a $\pi(A)$-module for the Jordan multiplication,  whose elements anticommute with $\chi$.
\end{enumerate}
\end{definition}

Just as in section \ref{SectionAssoc}, a Dirac operator for $\mathcal{B}$ is an operator $D$ such that $(A,H,\pi,D,\chi,J)$ is a Jordan triple and $\Omega^1_D A\subset \Omega^1$. The configuration space and automorphisms are also defined exactly as in the associative case. 

Condition $C_0$ has a very important interpretation in the Jordan setting. 
\begin{proposition}\label{symact}   The `symmetrized' action $S 
=\frac{1}{2}(\pi + \pi^\circ)$ is  a multiplicative representation of $A$ on $H$.
\end{proposition}
\begin{proof}
	We already know that $\pi$ is an associative representation. Let us prove that $\pi^o$ also is: for all $a,b\in A$ we have $\pi^o(a\circ b)=\pi(a\circ b)^o=(\pi(a)\circ \pi(b))^o=\pi^o(b)\circ \pi^o(a)=\pi^o(a)\circ\pi^o(b)$ since $\circ$ is commutative. Now $\pi$ and $\pi^o$ satisfy the hypotheses of Theorem \ref{JacTh14} by $C_0$, and the proposition follows.
\end{proof}

Condition $C_1$ is  generalized without change, it reads $[\pi(a),\pi(\omega)^o]=0$ or equivalently $[\pi(a)^o,\pi(\omega)]=0$ for all $a\in A$ and $\omega\in \Omega^1$. Unless specified otherwise, we will generally prefer to use the infinitesimal version of condition weak $C_1$:  

\begin{equation}\tag{weak $C_1$}
[[\pi(A),\pi(A)]^o,\Omega^1]\subset\Omega^1, 
\end{equation}

We will see in proposition \ref{Usub} below that {weak $C_1$} keeps the same meaning as in the associative case.  

We will have more to say  about the manifestations of the order $0$ and $1$ conditions in the Jordan setting in section \ref{ordersymmetries}, but for now let us take a closer look at the manifold case. We suppose $M$ to be parallelizable and we define the canonical Jordan background over $M$ to be the same as in the associative setting, which makes sense since the algebra $A=\mathcal{C}^\infty(M)$ is associative. Let us quickly recall the construction. We pick a moving frame $(e_a)$. It defines a metric $g_0$ and a spin structure at the same time. In particular it defines the Clifford mapping $\gamma$. Note that $g_0$ is not a background structure since it is not fixed by the automorphisms. The definition of $J_M$, $\chi_M$ and $\Omega^1_M$ are the same as usual \cite{AB1}. In particular

\begin{equation}
\Omega^1_M=\{ i\gamma(\alpha)|\alpha\text{ is a real smooth 1-form on }M\}.\label{Omega1man}
\end{equation}

and  $D_0$, the canonical Dirac operator associated with $e_0$ is a regular Dirac operator.  Note that, since $J_M$ anticommutes with $i$ and with Clifford elements of odd degree, it commutes with the elements of $\Omega^1_M$, and since those are anti-selfadjoint, they are self-opposite, a fact which will be useful in section \ref{AAsec}.

%%%%%%%%% this section is not needed anymore  
%\section{Upsilon}
%Consider a Jordan background or triple. For any $T\in B(H)$  define 
%
%\begin{equation}
%\Upsilon(T):=T(T^\dagger)^{o}=TJTJ^{-1}
%\end{equation}
%
%Thanks to C0, the map $\Upsilon$   satisfies for all $S,T\in B(H)$
%
%\begin{enumerate}
%\item  $[J,\Upsilon(T)]=0$,
%\item $\Upsilon(ST)=\Upsilon(S)\Upsilon(T)$,
%\item $\Upsilon(T^\dagger)=\Upsilon(T)^\dagger$.
%\end{enumerate}
% 
%In particular $\Upsilon$ defines a group homomorphism from $U(H)$ to $U(H)$.
%
%If $g(t)$ is a 1-parameter group of invertible operators, then $\Upsilon(t):=\Upsilon(g(t))$ has derivative $\Upsilon'(0)=g'(0)+Jg'(0)J^{-1}$, do that the  Lie algebra homorphism $d_e\Upsilon : B(H)\rightarrow B(H)$ is   defined by
%
%\begin{equation}
%d_e\Upsilon(h)=h+JhJ^{-1} 
%\end{equation}
%
%If $g\in U(H)$ and $h\in B(H)_{sa}$, then :
%\begin{itemize}
%\item $\Upsilon(g)=g(g^{-1})^o$, $d_e\Upsilon(ih)=ih-(ih)^o=i(h-h^o)$.
%\item $\Upsilon(h)=hh^o$, $d_e\Upsilon(h)=h+h^o$.
%\end{itemize}

\section{Lifted inner automorphisms and minimal fluctuations}\label{fluctuationsection}

In this section we seek to define a fluctuation space ${\cal F}_D$ for Jordan spectral triples. In order to do this, we first guess the correct algebraic structure of this fluctuation space, basing ourselves on 3 first principles. As we have recalled in section \ref{SectionAssoc}, in the associative case the fluctuation space is a certain affine space containing pure gauge flucutations \eqref{AssocInnerFluct1} and which is stable under automorphisms. Pure gauge fluctuations are associated to inner automorphisms of the algebra, though not in a 1-1 way: this is the reason why the left arrows in diagram \eqref{assocliftdiag} go in the wrong direction.
%In the usual $C^*$-situation, particular fluctuations (pure gauge) are given by $F_u:=\Upsilon(u)D\Upsilon(u)^{-1}-D$ where $u$ is a unitary element of $A$. The idea is that $\Upsilon$ lifts the inner automorphisms of $A$ to  unitary operator commuting with $\chi$ and $J$. Let's see if we can copy that in the Jordan case. 
%\begin{rem}\label{remunimod}
%$\Upsilon$ does not directly lifts  the inner automorphisms in the $C^*$-case, at least when the representation of the algebra is just $\RR$-linear as for the SM. For instance if $u=(e^{i\theta},1,1)\in {\mathbb C}\oplus {\mathbb H}\oplus M_3({\mathbb C})$, $Ad_u=1$ but $\Upsilon(u)\not=1$. This is the origin of the unimodularity problem. 
%\end{rem}
%
%The remark above is wrong ! This is in fact not a problem but the reason why U(1) symmetry group  
%are present in the associative case while they cannot be obtained (at least not so easily) in the 
%Jordan case. A different problem !
In the Jordan case, we can actually do better thanks to diagram \eqref{diagraminversion} and obtain the following:

\begin{equation}
\xymatrix{{\mathrm Aut}_{\mathrm Inn}(\pi(A))\ar@{~>}[r]& U(A)\ar[r]^{\Upsilon} &  {\mathcal U}(A) \cr
{\mathrm Der}_{\mathrm Inn}(\pi(A))\ar[u]^{\exp}\ar[r]^{a}& [\pi(A),\pi(A)]\ar[u]^{\exp} \ar[r]^{d_e\Upsilon}& {\mathcal G}(A)\ar[u]^{\exp}
}\label{liftdiag}
\end{equation}

The notations are the same as in section \ref{SectionAssoc}, namely  $\mathcal{U}(A):=\Upsilon(U(A))$ and $\mathcal{G}(A):=d_e\Upsilon([\pi(A),\pi(A)])$. For the sake of clarity, let us follow the bottom line in detail. We start with    $\delta\in\mathrm{Der}_{\mathrm{Inn}}(\pi(A))$. It is of the form $\delta=\mathrm{ad}(a_\delta)$ where

\begin{equation}
a_\delta= \frac{1}{4}\sum_i[\pi(a_i),\pi(b_i)]\label{extderiv}
\end{equation}

is a uniquely defined element of $[\pi(A),\pi(A)]$ thanks to proposition \ref{invertibility}. Then $d_e\Upsilon(a_\delta)=a_\delta+Ja_\delta J^{-1}=a_\delta-a_\delta^o$. Now observe that $a_\delta-a_\delta^o=\sum_i [S(a_i),S(b_i)]$ by $C_0$. Hence we obtain:

\begin{proposition}\label{GS} One has $\mathcal{G}(A)=[S(A),S(A)]$.
\end{proposition}
 
%Proposition \ref{GS} might be a key to go beyond the special case, see section \ref{secbeyond}.

\begin{proposition}\label{Usub} Under weak $C_1$, $\mathcal{U}(A)$ is a subgroup of the automorphism group of $\mathcal{B}$. 
\end{proposition}
\begin{proof}
Let $\Upsilon(u)\in \mathcal{U}(A)$. It is clear that it is unitary and commutes with $\chi$ and $J$. Moreover, $\mathrm{Ad}_{\Upsilon(u)}(\pi(A))=\mathrm{Ad}(u)(\pi(A))=\pi(A)$ by $C_0$. Finally, let us write $u=\exp(x_k)\ldots\exp(x_1)$, $x_1,\ldots,x_k\in [\pi(A),\pi(A)]$, and let $\omega\in\Omega^1$. Then $\mathrm{Ad}_{\Upsilon(u)}(\omega)=\mathrm{Ad}_{(u^{-1})^o}(\mathrm{Ad}_u(\omega))$ by $C_0$. Now
\begin{eqnarray}
\mathrm{Ad}_u(\omega)&=&\mathrm{Ad}_{\exp(x_k)}(\ldots(\mathrm{Ad}_{\exp(x_1)}(\omega)\ldots)\cr
&=&\exp(\mathrm{ad}_{x_k}(\ldots(\exp(\mathrm{ad}_{x_1}(\omega)\ldots)
\end{eqnarray}
Since $\Omega^1$ is stable under the adjoint action of $[\pi(A),\pi(A)]\subset\mathrm{Lie}_\pi(A)$ by definition, we have $\mathrm{Ad}_u(\omega)\in\Omega^1$. By weak $C_1$, $\Omega^1$ is also a $[\pi(A),\pi(A)]^o$-module, and we can repeat the same proof to show that $\mathrm{Ad}_{(u^{-1})^o}(\mathrm{Ad}_u(\omega))\in \Omega^1$.
\end{proof}

%\begin{rem} We see in the proof of the proposition that we need  $\Omega^1$  to be a $[\pi(A),\pi(A)]$-module, which is implied by definition \ref{defJordanForms} since $\mathrm{Lie}_\pi(A)=\pi(A)\oplus [\pi(A),\pi(A)]$.  If we had defined $\Omega^1$ to be the $[\pi(A),\pi(A)]$-submodule generated by $D$, we would have obtained  proposition \ref{Usub} from the minimal set of hypotheses, however $\Omega^1$ would have been trivial in the manifold case. It thus appears to definition \ref{defJordanForms} is optimal in a certain sense.
%\end{rem}

As can be seen from   proposition \ref{Usub}, diagram \eqref{liftdiag} can be completed as the commutative cube:
\begin{equation}
\xymatrix{
& {\mathrm Aut}_{\mathrm Inn}(\pi(A)) 
& & {\mathcal U}(A)\ar[ll]^{\Ad}
\\
{\mathrm Aut}_{\mathrm Inn}(\pi(A))\ar[ur]^{\mathrm{Id}}  \ar@{~>}[rr]
& & U(A) \ar[ur]_{\Upsilon}  
\\
& {\mathrm Der}_{\mathrm Inn}(\pi(A))\ar@{-->}[uu]^(0.3)\exp
& & [S(A),S(A)]\ar@{-->}[ll]_{\ad}\ar[uu]^(0.3)\exp
\\
{\mathrm Der}_{\mathrm Inn}(\pi(A))\ar[uu]_(0.3){\exp} \ar[rr]^{a}\ar@{-->}[ur]^{\mathrm{Id}}& & [\pi(A),\pi(A)] \ar[uu]^(0.3){\exp}\ar[ur]^{d_e\Upsilon}
}
\label{liftdiagcompleted}
\end{equation}
For instance, the path followed by the inner derivation $[L_{\pi(a)},L_{\pi(b)}]$ is: 
\begin{equation}
[L_{\pi(a)},L_{\pi(b)}] \mapsto \frac{1}{4}[\pi(a),\pi(b)]\mapsto \frac{1}{4}[\pi(a),\pi(b)]- \frac{1}{4}[\pi(a),\pi(b)]^o\mapsto\ad_{ \frac{1}{4}[\pi(a),\pi(b)]- \frac{1}{4}[\pi(a),\pi(b)]^o}=[L_{\pi(a)},L_{\pi(b)}]\label{path}
\end{equation}
 
Note that $\Ad$ and $\ad$ could as well be replaced with $\Ad^*$ and $\ad^*$.

\begin{rem} When proposition  \ref{invertibility} holds, all the arrows in the bottom face of the cube \eqref{liftdiagcompleted} are isomorphisms, in particular $d_e\Upsilon$. It means that there cannot be non-trivial self-opposite elements in $[\pi(A),\pi(A)]$. Note that the situation for associative backgrounds is much more involved (see \cite{AB2}).
\end{rem}

Now let us come to the fluctuation space ${\mathcal F}_D$. We want to guess what this space is basing ourselves on the following postulates:

\begin{enumerate}
\item It is a real vector space.
\item It contains the pure gauge fluctuations $UDU^{-1}-D$ for all $U\in {\mathcal U}(A)$.
\item\label{fluct3} For all $F\in {\mathcal F}_D$, ${\mathcal F}_{D+F}\subset {\mathcal F}_D$.
\end{enumerate}

The third postulate says that ``a fluctuation of a fluctuated Dirac is a fluctuation of the original Dirac''. This property holds in the associative case and seems desirable to maintain. However we could  think of requiring instead  a minimality property like the following one:

\begin{itemize}
\item[3'.] ${\mathcal F}_D$ is generated as a vector space by 2.
\end{itemize} 

Later on we will see that the system $1,2,3'$ is actually stronger than $1,2,3$. But first we prove the following result:

\begin{proposition}\label{propo7} If ${\mathcal F}_D$ satisfies either $1,2,3$ or $1,2,3'$, then it satisfies
\begin{itemize}
\item[4.] for all $F\in\mathcal{F}_D$ and $U\in\mathcal{U}(A)$, $Ad_U(F)\in\mathcal{F}_D$.
\end{itemize} 
\end{proposition}
\begin{proof}
By 2, $Ad_{UV}(D)-D$ and $Ad_U(D)-D$ are both in ${\mathcal F}_D$, hence their difference $Ad_U(Ad_V(D)-D)\in\mathcal{F}_D$ by 1. By linearity this shows that $1,2,3'$ implies $4$.

Now let $F\in \mathcal{F}_D$. Then   $Ad_U(D+F)-(D+F)\in \mathcal{F}_{D+F}$ by 2. If $3$ holds, it is also in $\mathcal{F}_D$. Thus $(Ad_U(D)-D)+Ad_U(F)-F\in \mathcal{F}_D$. But $Ad_U(D)-D\in \mathcal{F}_D$ by $2$ and $F\in\mathcal{F}_D$ by hypothesis. Hence $Ad_U(F)\in\mathcal{F}_D$ by $1$.
\end{proof}

\begin{proposition} $1,2,3'\Rightarrow 3$.
\end{proposition}
\begin{proof}
Let $F\in\mathcal{F}_D$. By $3'$ a general element of $\mathcal{F}_{D+F}$ is a linear combination of $Ad_{U_i}(D+F)-(D+F)=(Ad_{U_i}(D)-D)+Ad_{U_i}(F)-F$. The first summand is in $\mathcal{F}_D$ by $2$, the third one by hypothesis, and the middle one by the previous proposition. Hence by $1$, $\mathcal{F}_{D+F}\subset\mathcal{F}_D$.
\end{proof}

In this section we will investigate a minimal fluctuation space consistent  with the stronger set of axioms  $1,2,3'$. In the next section we will consider more general fluctuations consistent with $1,2,3$ for the case in which $C_1$ holds. Proposition \ref{propo7} tells us that  ${\mathcal F}_D$ is an invariant space for the adjoint representation of $U(A)$ on $B(H)$. If $\mathcal{F}_D$ is closed we obtain by  differentiation  that  ${\mathcal F}_D$ is a module over the Lie algebra $\mathcal{G}(A)=[S(A),S(A)]$ for the action $\mathrm{ad}$. Moreover it contains $ad_h(D)$ for all $h\in [S(A),S(A)]$. All of this suggest the following possible definition:

\begin{definition}\label{qwak}   The minimal fluctuation space $\mathcal{F}_D$ is the Lie module over  $\mathcal{G}(A)$ generated by elements of the form $[\delta,D]$, for $\delta\in \mathcal{G}(A)$.

\end{definition}

%	  $\mathcal{F}_D:=[[T,D]]_{[S(A),S(A)]}=[D]_{\mathcal{G}(A)}$, for $T\in [S(A),S(A)]$, where .  

%We define $\mathrm{Lie}_\pi(A)$ as in the preliminaries. We will use the following notation: if $\mathcal{G}$ is a Lie algebra acting on vector space $V$ and $X\subset V$, then $[X]_\mathcal{G}$ will denote the   smallest closed submodule of $V$ containing $X$. In what follows $V$ will always be $\mathrm{End}(H)$.

\begin{rem}
Since $\pi$ is faithful, the fluctuation space $\mathcal{F}_D$ is also a Lie module over inner derivations of $A$. 
%More precisely, $\mathcal{F}_D=[D]_{\mathrm{Der}_\mathrm{Inn}(A)}$ for the action $\delta\cdot D=\ad_{d_e\Upsilon(a(\delta))}(D)$.
\end{rem}

Conversely, it is easy to see that with   definition \ref{qwak}, %\ref{defFluctSpace},
 $1,2,3$ hold.   First, property 1 holds by definition. Let us prove 2. If $U=\exp(h)$ with $h\in\mathcal{G}(A)$, then $UDU^{-1}-D=\sum_{k=1}^\infty\frac{1}{k!}\mathrm{ad}^k_h(D)\in \mathcal{F}_D$. Suppose we have proved the result for any $U$ which is a product of $n-1$ exponentials and let $V=\exp(h)U$. Then 
\begin{eqnarray}
VDV^{-1}-D&=&Ad_{\exp(h)}(UDU^{-1}-D)+Ad_{\exp(h)}(D)-D\cr
&=&\sum_{k=1}^\infty \frac{1}{k!}\mathrm{ad}^k_h(UDU^{-1}-D)+UDU^{-1}-D+Ad_{\exp(h)}(D)-D\in \mathcal{F}_D
\end{eqnarray}
Property 2 follows by induction. Now let us prove 3. We let $F\in\mathcal{F}_D$. Then $ad_h(D+F)=ad_h(D)+ad_h(F)\in\mathcal{F}_D$. It follows that the submodule generated by $D+F$ is a subset of $\mathcal{F}_D$.

%\begin{rem} In infinite dimension we must take the closure of the submodule. Moreover, I am not familiar with infinite-dimensional Lie groups, hence I don't know if finite products of exponentials suffice. This is probably OK for AC manifolds with compact smooth part.
%\end{rem}

Now that we have a well-motivated definition for the minimal fluctuation space, let us take a closer look  at the form they take. It will be a linear combination of terms $\delta_k\cdot\ldots\cdot \delta_1\cdot D$, for $k\ge 1$, where $\cdot$ means the adjoint action and $\delta_j=T_j-T_j^o$ with $T_j\in [\pi(A),\pi(A)]$.  

\begin{proposition}\label{subspaceconfig} If weak $C_1$ holds,  $D+\mathcal{F}_D$ is a subspace of the configuration space.
\end{proposition}
\begin{proof} We must prove that each $D+F$ with $F\in\mathcal{F}_D$ is a Dirac operator.  All the required properties are obviously satisfied except $\{[D+F,\pi(a)],a\in A\}\subset \Omega^1$. To see that it also  holds under weak $C_1$, it suffices to show that for all $a\in A$, $[F,\pi(a)]\in \Omega^1$. We can suppose without loss of generality that $F$ is of the form $F=\delta_k\cdot\ldots\cdot \delta_1\cdot D$, with  $\delta_j=T_j-T_j^o$ for $k\ge 1$, and $T_j\in [\pi(A),\pi(A)]$, as above. By linearity we only need to consider the case $F=R_k\cdot\ldots\cdot R_1\cdot D$ with $R_j\in[\pi(A),\pi(A)]$ or $R_j\in [\pi(A),\pi(A)]^o$, with $j=1,\ldots,k$.    If $k=1$ then $[\pi(a),F]=\pi(a)\cdot R_1\cdot D=(\pi(a)\cdot R_1)\cdot D+R_1\cdot(\pi(a)\cdot D)$ by Jacobi's identity. Now $\pi(a)\cdot R_1$ is   an element of $\pi(A)$ if $R_1\in[\pi(A),\pi(A)]$ since $\mathrm{ad}_{R_1}$ is a derivation of $\pi(A)$, and it vanishes by $C_0$ if $R_1\in[\pi(A),\pi(A)]^o$. On the other hand $\pi(a)\cdot D\in\Omega^1$, so that $R_1\cdot (\pi(a)\cdot D)$ also belongs to $\Omega^1$ using the fact that $\Omega^1$ is a both a $[\pi(A),\pi(A)]$ and a $[\pi(A),\pi(A)]^o$-module. Now suppose the property is proved for some $k$. By Jacobi again we have $a\cdot R_{k+1}\cdot\ldots\cdot R_1\cdot D=\sum_{j=1}^{k+1}R_{k+1}\cdot \ldots (a\cdot R_j)\ldots \cdot R_1\cdot D$. If $R_j\in[\pi(A),\pi(A)]^o$, the summand vanishes by $C_0$. If $R_j\in[\pi(A),\pi(A)]$ then $a\cdot R_j\in\pi(A)$, and $(a\cdot R_j)\ldots \cdot R_1\cdot D\in \Omega^1$ by induction. Then the summand belongs to $\Omega^1$ since $\Omega^1$ is a $[\pi(A),\pi(A)]$ and a $[\pi(A),\pi(A)]^o$-module.
\end{proof}

\begin{rem} It is remarkable that proposition \ref{subspaceconfig} holds under weak $C_1$ alone in the Jordan case, whereas an additional condition called weak $C_1'$ is needed in the associative case \cite{weakC1}.
\end{rem}

We have seen that definition \ref{defFluctSpace} entails properties $1,2,3$, and   from proposition \ref{propo7} it follows that $\mathcal{F}_D$ is always automorphism invariant. Now the same will be true for $D+\mathcal{F}_D$ since $U(D+F)U^{-1}=D+(UDU^{-1}-D)+UFU^{-1}$. Thus, under  weak $C_1$,   $D+\mathcal{F}_D$ is an automorphism invariant subspace of the configuration space. This means that  a particle model can consistently be defined on this space. 

%Let us end this section with some thoughts on alternative possibilities for the fluctuation space. 

 The fluctuation space given by definition \ref{defFluctSpace} is the smallest one which respects the principles we have set forth, among them gauge-invariance. Let us observe that our approach here differs somewhat from what is  usually done in the associative setting. The definition of the fluctuation space given by Connes is not only guided by the physically well-motivated principle of gauge-invariance, but by a deep generalization, namely Morita self-equivalence of spectral triples. Furthermore, Connes' fluctuation space is usually defined in the presence of $C_1$, with the story becoming somewhat more complicated  in the absence of  $C_1$~\cite{walterc1}. As we are dealing with special Jordan triples satisfying weak $C_1$, and it is not clear what the analogue of Morita self-equivalence might be in the Jordan setting,   we opt  for gauge-invariance and minimality. An obvious question is what this same approach would yield in the associative setting. In the next section we consider the construction of more general fluctuations for special Jordan triples that satisfy  $C_1$.

\section{Order conditions and general fluctuations}\label{ordersymmetries}

In previous sections we limited our discussion of  the order conditions $C_0$ and $C_1$. In general we will not restrict attention to geometries satisfying $C_1$, however it is important to understand the implications  that these conditions have, as many physically relevant and interesting geometries  will satisfy both conditions. In this section we take a closer look at  special Jordan representations in the presence of both $C_0$ and $C_1$, focusing in particular on their symmetries. 

We begin with $C_0$. Consider a (real, even) special, Jordan, pre-spectral  triple $B = (A,H,\pi,\chi,J)$. The representations $\pi$ and $\pi^0$ both individually satisfy the properties of an associative specialization.  We have already seen (proposition \ref{symact}) that it is possible to  form a new `symmetrized' action $S =\frac{1}{2}(\pi + \pi^\circ)$   which satisfies all of the properties of a multiplicative representation. Following definition \ref{defjordanmodule}, we therefore  see that  $(H,S)$ is a Jordan module, or equivalently the pre-spectral triple can be viewed as a Jordan algebra $B =A\oplus H$, where the bilinear product  extends $\circ$ on $A$ such that $a\circ h=h\circ a=S_ah$, $h\circ h'=0$, for all $a\in A$, $h,h'\in H$. Note, that both $J$ and $\chi$ will commute with the `symmetrized' representation $S=\frac{1}{2}(\pi+\pi^\circ)$, a fact which has deep implications for the construction of  physical theories with Majorana fermions. In particular, notice that only representations that commute with $J$ are compatible with the Majorana condition $JH = H$. This is a key motivation that underlies the construction of the standard model as a Jordan geometry ~\cite{Boyle:2020}.

The meaning of $C_0$ is clear. Following theorem \ref{JacTh14}, its imposition ensures that the symmetrized action of $A$ on $H$ satisfies all of the properties of a Jordan action. The upshot  is that for a `symmetric' representation satisfying $C_0$, $B =A\oplus H$ will be a  Jordan algebra, and following definition \ref{innerdeff},  its inner derivations will  be of the form 
\begin{align}
\delta = \sum  [L_a,L_b],\label{derivation0}
\end{align}
where $a,b\in B$  and $L_a$ is the Jordan multiplication by $a$ when acting on $A$ and is equal to $S_a$ when acting on $H$. Notice, however, that when either $a$ or $b$ is drawn from $H$, we have $\delta_{a,b}h=0$ for all $h\in H$. Following  definition \ref{iinnerauts}, the inner automorphisms of $B$ acting on $H$ will therefore be of the form $\alpha = e^\delta\in \mathcal{U}(A)$, where $\delta$ is  constructed from elements from $A$, and  not from $B=A+H$ more generally, i.e. $\in \delta\in \mathcal{G}(A)$\footnote{The physical meaning of the derivation of the form $\delta_{ab}$, in which either $a$ or $b$ are elements in $H$, is an interesting question, which is outside the scope of this paper.}. Inner automorphisms of $B$ acting on $H$ will therefore automatically commute with both $J$ and $\chi$ (because the symmetric action on $H$ commutes with both $J$ and $\chi$). In other words, $C_0$ ensures that $S$ is multiplicative, and as such that  the symmetries of $A$ can be `lifted' in a consistent way to the symmetrized representation of $A$ on $H$.

Next, let's consider the meaning of $C_1$.  We restrict attention to expressions of degree one and lower, meaning  we will not  consider the representation of products of forms on $H$. The representation of higher order forms is an involved discussion even in the associative setting, and deserves a paper in its own right. As a representation of higher order forms will not be necessary for deriving any of the results in this paper, we will return to the discussion  in a follow-up paper where we will show that Jordan geometries are somewhat better behaved than associative geometries at higher order (see, however, \cite{Brouder:2015qoa,Boyle:2014wba} for a more involved discussion of Junk forms and the second order condition in the associative setting).

 We begin by  equipping the  pre-spectral triple $B$ with a Dirac operator $D$, to form a special Jordan  triple $\mathcal{T} =  (A,H,\pi,D,\chi,J)$, that satisfies  all of the usual properties.  We observe that we have a map $\pi:  A\oplus \Omega^1_d  A\rightarrow  \mathrm{End}(H)$ such that $\pi$ is an associative representation on $A$ and a module map on $\Omega^1_d  A$. {  This} means that $\pi$ is an associative representation of the split null extension $A\oplus \Omega^1_dA$ \emph{up to degree $1$}. Similarly $\pi^o$ is an associative representation of   $A\oplus \Omega^1_dA$  up to degree $1$, and thanks to $C_0$ and $C_1$, these two ``representations'' commute. Thus, by theorem \ref{JacTh14}, $S = \frac{1}{2}(\pi +{\pi}^\circ)$, will  be a   multiplicative representation up to degree $1$. Since one might feel uncomfortable using this theorem ``up to degree one'', we provide a formal proof below.

\begin{proposition}\label{order1prop} The  `symmetrized' action of forms $S = \frac{1}{2}(\pi +  {\pi}^\circ)$ defined on $A\oplus \Omega^1_dA$ satifies the properties of a  multiplicative specialization for all expressions    up to degree 1. 
\end{proposition}

\begin{proof} 
	 We need to show that { the linearization of  equations \eqref{multrep}  given in} equations \eqref{equivmultrep}  hold for the case in which a single element $a,b,c\in \Omega A$ is of degree 1, and the remaining elements are of degree 0. We begin with the first equation: 
	\begin{align}
	[S_a,S_{b\circ \omega}]+[S_b,S_{a\circ \omega }]
	+[S_\omega,S_{b\circ a}] &=-\frac{1}{4} [\pi(a),\pi(b\circ \omega)^\circ]+\frac{1}{4}[\pi(a)^\circ,\pi(b\circ \omega)]\nonumber\\
	&\phantom{=} - \frac{1}{4}[\pi(b),\pi(a\circ \omega)^\circ]+\frac{1}{4}[\pi(b)^\circ,\pi(a\circ \omega)]\nonumber\\
	&\phantom{=}+ \frac{1}{4}[\pi(\omega),\pi(a\circ b)^\circ]-\frac{1}{4}[\pi(\omega)^\circ,\pi(a\circ b)]\nonumber\\
	&=  0,
	\end{align}
	for $a,b\in A$, $\omega\in \Omega^1 A$, where the first equality holds because {  as $\pi$ and $\pi^\circ$ are associative specializations}, $\frac{1}{2}\pi$ and $\frac{1}{2}\pi^\circ$ satisfy the equation separately. The second equality holds due to  $C_1$. {  For the second equation we find }
	\begin{align}
	S_aS_bS_\omega+S_\omega S_bS_a + S_{(a\omega)b}-S_aS_{b\omega}-S_{b}S_{a\omega}-S_{\omega}S_{ab}&=\frac{1}{8}[\pi(a),\pi(b)^\circ]\pi(\omega) +\frac{1}{8}[\pi(\omega),\pi(b)^\circ]\pi(a)\nonumber\\
	&+\frac{1}{8}[\pi(b),\pi(\omega)^\circ]\pi(a)^\circ +\frac{1}{8}[\pi(b),\pi(a)^\circ]\pi(\omega)^\circ\nonumber\\
	&+\frac{1}{8}[\pi(\omega)^\circ,\pi(a)\pi(b)]+\frac{1}{8}[\pi(\omega)\pi(b),\pi(a)^\circ]\nonumber\\
	&+\frac{1}{8}[\pi(a)^\circ\pi(b)^\circ,\pi(\omega)]+\frac{1}{8}[\pi(a),\pi(\omega)^\circ\pi(b)^\circ]\nonumber\\
	&=0,
	\end{align}
	for $a,b\in A$, and $\omega\in \Omega^1 A$,	where the first equality holds because  $\frac{1}{2}\pi$ and $\frac{1}{2}\pi^\circ$ satisfy {  the equation} separately, and the second equality holds due to  $C_1$. Similarly
	\begin{align}
	S_aS_\omega S_c+S_c S_\omega S_a + S_{(a c)\omega}-S_aS_{\omega c}-S_{\omega }S_{ac}-S_{c}S_{a \omega}&=\frac{1}{8}[\pi(\omega)^\circ,\pi(a)]\pi(c) +\frac{1}{8}[\pi(\omega)^\circ,\pi(c)]\pi(a)\nonumber\\
	&+\frac{1}{8}[\pi(a)^\circ,\pi(\omega)]\pi(c)^\circ +\frac{1}{8}[\pi(c)^\circ,\pi(\omega)]\pi(a)^\circ\nonumber\\
	&+\frac{1}{8}[\pi(a),\pi(c)^\circ\pi(\omega)^\circ]+\frac{1}{8}[\pi(c),\pi(a)^\circ\pi(\omega)^\circ]\nonumber\\
	&+\frac{1}{8}[\pi(a)\pi(\omega),\pi(c)^\circ]+\frac{1}{8}[\pi(c)\pi(\omega),\pi(a)^\circ]\nonumber\\
	&=0,
	\end{align}
	for $a,c\in A$, $\omega\in \Omega^1 A$. 
\end{proof}

For special Jordan triples that satisfy $C_1$, the symmetrized action of $A$ on $H$ therefore satisfies the properties of a multiplicative specialization for all expressions of degree  $1$ and lower. In effect, just as $C_0$ extends the Jordan product $\circ$ on $A$ to all of $B=A\oplus H$ as a Jordan product, the first order condition $C_1$ extends the product further  as a Jordan product to  $\mathcal{T} = A\oplus \Omega^1_DA\oplus H$, so long as one restricts attention to expression of degree  $1$ or lower. The upshot is that for special Jordan triples satisfying $C_1$, we are able to extend the `degree zero' inner derivations on $B=A\oplus H$ to include `degree one' elements $\delta:A\oplus H\rightarrow \Omega^1_d\oplus H$, of the form
\begin{align}
  \sum[L_a,L_\omega], \label{d1ders}
\end{align}
for $a\in A$, $\omega\in \Omega^1_d A$\footnote{In this paper we will not consider derivations of the form $\delta_{\omega h} = [L_\omega,L_h]$ for $h\in H$  as they act trivially on $H$.}.  Following the same route as in \eqref{path} yields the operator:

\begin{align}
F = \sum [S_a,S_\omega].\label{flucgen}
\end{align}

 Let it be clear that a rigorous justification of this step would require to define the full algebra of forms, the extension of $A$ by this algebra (in the spirit of \cite{Boyle:2014wba}) and the upgrading of \eqref{liftdiagcompleted} to this new setting, which is beyond the scope of this paper. Here we will be happy to observe that things are working at first order, and take it as a motivation to consider operators of the form \eqref{flucgen}.
 It is easy to show that  they satisfy the Leibniz rule and in particular, following Eq.~\eqref{order0}:
\begin{align}
[[S_a,S_\omega],S_b] = S_{[\omega,b,a]}\label{order1lieb}
\end{align}
for all $a,b\in A$, and $\omega\in \Omega^1_d A$. We denote the space of  degree one derivation elements { of the form given in \eqref{flucgen}}
  by ${\mathcal F}_\omega$. Notice, that unlike at  degree zero, these derivations are self-adjoint, commute with $J$, and anti-commute with $\chi$. 

We now turn to the discussion of Fluctuated Dirac operators in the presence of $C_1$. 

\begin{definition}\label{defFluctSpace} Given a special Jordan Triple $\mathcal{T}=(A,H,\pi,D,\chi,J)$,  we define the general fluctuation space to be given by $\mathcal{F}_\omega$.
\end{definition}

In addition to having Hermitian elements that commute with $J$, anti-commute with $\chi$, and map zero forms to one forms through commutation, the fluctuation space $\mathcal{F}_\omega$ satisfies the three postulates that we set out in the preceding section. In particular $\mathcal{F}_\omega$ is a real vector space that contains $\mathcal{F}_\mathcal{D}$, and for all $F\in {\mathcal F}_\omega$, ${\mathcal F}_{D+F}\subset {\mathcal F}_D$. 

\begin{proof}
$\mathcal{F}_\omega$ is a real vector space by definition, so we focus on the other two postulates. 
For the second postulate, we have only to  show that elements of the form $\delta_k\cdot\ldots\cdot \delta_1\cdot D$, for $k\ge 1$,  are in $\mathcal{F}_\omega$. To begin with, for $k= 1$ we have
\begin{align}
[D,[S_a,S_b]] = [S_{d[a]},S_b] + [S_{a},S_{d[b]}]\in \mathcal{F}_\omega, 
\end{align} 
for $a,b\in A$, which follows directly from Eq.~\eqref{da} and the definition of the symmetrized action.  Moreover, for all $F\in \mathcal{F}_\omega$, one has $[[S_a,S_b],F]\in \mathcal{F}_\omega$ by Jacobi's identity and Eq.~\eqref{order1lieb}. It then follows  that   $\mathcal{F}_\omega$ is a Lie module over $ \mathcal{G}(A)$, and as a result all elements of the form $\delta_k\cdot\ldots\cdot \delta_1\cdot D$ are in $\mathcal{F}_\omega$, proving the second postulate. Similarly,  postulate $3$ follows directly from Eq.~\eqref{order1lieb}.
\end{proof}

\begin{rem}
Finally, before closing this section, we make a brief comparison between $\mathcal{F}_\omega$, and the general associative fluctuations given in Eq.~\eqref{DefFluctAssoc}. { 
 The analogue of equation \eqref{derivation0} in the associative setting is given by:
\begin{align}
\delta =\sum  L_a - R_a
\end{align}
where $a\in B$, and where for $a\in \cal{A}$ the `left action' is given by $L_aa' = aa'$ and $L_ah =\pi(a)h$, while the `right action' is given by $R_aa' = a'a$ and $R_ah =\pi^\circ(a)h$, for $a'\in \cal{A}$ and $h\in H$. For associative geometries satisfying $C_1$, the `degree zero' inner derivations on $B=A\oplus H$ can then be extended to include `degree one' elements $\delta:A\oplus H\rightarrow \Omega^1_d\oplus H$, of the form
\begin{align}
\sum L_\omega - R_\omega, \label{d1Jders}
\end{align}
for $\omega\in \Omega^1_d A$. The analogue of Eq~\eqref{flucgen} is then given by:
\begin{align}
\sum \omega+\omega^\circ.
\end{align}
Restricting to Hermitian elements of this form, we obtain Connes' fluctuations Eq.~\eqref{DefFluctAssoc}. Notice that in the Jordan setting we obtain Hermiticity for free, it is not put in by hand.

 Further, comparison is able to be made between the associative and Jordan settings, by expressing Eq.~\eqref{flucgen}} more explicitly in terms of the associative representations $\pi$ and $\pi^\circ$. 
\begin{align}
[S_a,S_\omega] &=[\pi(a)+\pi(a)^\circ,\pi(\omega)-\pi(\omega)^\circ]\nonumber\\
&=[\pi(a),\pi(\omega)]+J[\pi(a),\pi(\omega)]J^{-1}.\label{genfluctC1}
\end{align}
We see that the general fluctuation space of a Jordan geometry is slightly more restrictive than the corresponding fluctuation space for an associative geometry, since $[\pi(a),\pi(\omega)]$ is a traceless $1$-form. Hence unimodularity is an automatic feature of the general fluctuation space. A curious question, is what phenomenological restrictions on the scalar sector of particle theories these additional restrictions will bring.
\end{rem}

% If $C_1$ holds, then \emph{weak $C_1$} will automatically be satisfied. This is  because $C_1$ ensures that the action of $\Omega^1$ on $H$ associates with the action of $A$ on $H$~\cite{Boyle:2014wba}. This, together with the associative bi-module structure of $\Omega^1$ ensures that the associative automorphisms  of the underlying pre-spectral triple $B$ will be compatible with the action of $\Omega^1$ on $H$. 

\section{Jordan 1-forms and fluctuations for almost-associative special Jordan triples}\label{AAsec}
The tensor product of two Jordan algebras with product $(a\otimes b)\circ(c\otimes d)=a\circ b\otimes c\circ d$ is generally not a Jordan algebra.  However, this works if at least one of the algebras is associative. Here we will consider algebras of the form 
\begin{equation}
A=\mathcal{C}^\infty(M)\otimes A_F=\mathcal{C}(M,A_F)\label{AlmostAssAlg}
\end{equation}
where $M$ is a manifold and $A_F$ is a finite-dimensional Jordan algebra.  We call these algebras \emph{almost-associative} by analogy with almost-commutative ones.

%First we observe that if $A$ is of the form \eqref{AlmostAssAlg}, then $\mathrm{Lie}_\pi(A)=\mathcal{C}^\infty(M)\otimes \mathrm{Lie}_\pi(A_F)$.
 
Let $\mathcal{B}_M$ be the canonical Jordan background over $M$ and $\mathcal{B}_F$ be a finite Jordan background. The definition of the almost-associative Jordan background $\mathcal{B}_M\hat\otimes \mathcal{B}_F$ is  the same as in the almost-commutative case. The algebra, real structure, and chirality operators are graded tensor products and follow the same rules as given in \cite{BiziBrouderBesnard} {  and~\cite{Farnsworth:2016qbp}}. In particular one has $(T_1\hat\otimes T_2)^o=(-1)^{|T_1||T_2|}T_1^o\hat\otimes T_2^o$, where $|T_{1,2}|$ is the grading of the corresponding operator, defined by its commutation property with $\chi_{1,2}$. All the necessary checks are exactly the same for the almost-associative and almost-commutative cases, except for the module of $1$-forms, which is given by

\begin{equation}
\Omega^1_{M\times F}= \Omega^1_{M}\otimes  \pi(A_F)\oplus \mathcal{C}^\infty(M)\otimes \Omega^1_{F}\label{defomegaprod}
\end{equation}

It is immediate to check that equation \eqref{defomegaprod} defines an odd Jordan $\pi(A)$-module. Moreover, 
let $D$ be a \emph{product Dirac} 
\begin{equation}
D=D_M\hat\otimes 1+1\hat\otimes D_F\label{productDirac}
\end{equation}
where $D_M$ is in the configuration space of $\mathcal{B}_M$ and $D_F$ is in the configuration space of $\mathcal{B}_F$. Then $\Omega^1_D{A}\subset \Omega^1_{M\times F}$, with equality whenever $D_F$ is regular.

Let us turn to fluctuations. We know that $\mathcal{F}_D$ is the Lie module generated by the orbit of $D$ under the action of $[S(A),S(A)]$. Let $A=\mathcal{C}^\infty(M,A_F)$, and $D$ be a product Dirac as in \eqref{productDirac}. We will need the following lemma. We recall that a Lie algebra $\mathcal{G}$ is called \emph{perfect} if $\mathcal{G}=[\mathcal{G},\mathcal{G}]$. Semisimple Lie algebras are perfect.

\begin{lemma} Consider a finite Jordan triple over   $A_F$ and its fluctutation space $\mathcal{F}_{D_F}$. Suppose $[S(A_F),S(A_F)]$ is a perfect Lie algebra and let $\mathcal{F}_{D_F}'$ be the \emph{derived fluctuation} space $\mathcal{F}_{D_F}':=[S(A_F),S(A_F)]\cdot \mathcal{F}_{D_F}$. Then    $\mathcal{F}_{D_F}'=\mathcal{F}_{D_F}$.
\end{lemma}
\begin{proof}
Since the inclusion $\subset$ is obvious, we only need to prove the converse. Every finite fluctutation is a sum of terms like $T_k\cdots T_1\cdot D_F$ with $T_i\in[S(A_F),S(A_F)]$. We only need to prove that a fluctuation of the form $T\cdot D_F$ with $T\in [S(A_F),S(A_F)]$ can be written as a sum of terms $T_k\cdots T_1\cdot D_F$ with $k\ge 2$. Since  $[S(A_F),S(A_F)]$ is perfect, we can write $T=\sum_i[\alpha_i,\beta_i]$, with $\alpha_i,\beta_i\in [S(A_F),S(A_F)]$. Now we have 
\begin{eqnarray}
T\cdot D_F&=&\sum_i [\alpha_i,\beta_i]\cdot D_F\cr
&=&\sum_i \alpha_i\cdot\beta_i\cdot D_F-\beta_i\cdot\alpha_i\cdot D_F,\mbox{ by Jacobi's identity}\cr
&\in\mathcal{F}_{D_F}'
\end{eqnarray}
\end{proof}

\begin{theorem}\label{thfluct}
Let $D=D_M\hat\otimes 1+1\hat\otimes D_F$ be the product Dirac operator of an almost-associative Jordan triple. Then 
\begin{equation}
\mathcal{F}_D\subset\Omega^1_M\otimes [S(A_F),S(A_F)]\oplus\mathcal{C}^\infty(M,\mathcal{F}_{D_F})\label{ACfluctu}
\end{equation}
with equality if $[S(A_F),S(A_F)]$ is a perfect Lie algebra.
\end{theorem}
\begin{proof}
Let us call $M$ RHS of \eqref{ACfluctu}. To prove that $\mathcal{F}_D\subset M$ it suffices to prove that $M$ is a $[S(A),S(A)]$-module which contains $[S(A),S(A)]\cdot D$. Let $f\otimes \alpha\in [S(A),S(A)]=\mathcal{C}^\infty(M,[S(A_F),S(A_F)])$. Then 
\begin{eqnarray}
(f\otimes\alpha)\cdot D&=&f\cdot D_M\otimes \alpha+f\otimes \alpha\cdot D_F\label{falphaonD}
\end{eqnarray}

  which belongs to the RHS of \eqref{ACfluctu}. Now let $\omega\in \Omega^1_M$, $\beta\in [S(A_F),S(A_F)]$, $g\in \mathcal{C}^\infty(M)$, and $\phi\in \mathcal{F}_{D_F}$. Then $f\otimes \alpha \cdot (\omega\otimes \beta+g\otimes \phi)=f\omega\otimes [\alpha,\beta]+fg\otimes \alpha\cdot \phi$ which also belongs in the right space since $[S(A_F),S(A_F)]$ is a Lie algebra and $\mathcal{F}_{D_F}$ an $[S(A_F),S(A_F)]$-module. 

Let us now prove the converse inclusion when $[S(A_F),S(A_F)]$ is a perfect Lie algebra. First, using $(1\otimes\alpha)\cdot D=1\otimes (\alpha\cdot D_F)$, we see that $1\otimes \mathcal{F}_{D_F}\subset \mathcal{F}_D$. Acting with $f\otimes \beta\in[S(A),S(A)]$ on $1\otimes \phi\in 1\otimes \mathcal{F}_{D_F}$ we obtain $f\otimes \beta\cdot\phi\in\mathcal{F}_{D}$, and we conclude that $\mathcal{C}^\infty(M,\mathcal{F}_{D_F}')\subset \mathcal{F}_D$. Using the lemma this shows that $\mathcal{C}^\infty(M,\mathcal{F}_{D_F})\subset \mathcal{F}_D$. There just remains to prove that $\Omega^1_M\otimes [S(A_F),S(A_F)]\subset \mathcal{F}_D$. From \eqref{falphaonD} we see that $f\cdot D_M\otimes \alpha$ is the difference of two elements of $\mathcal{F}_D$ and it thus also in $\mathcal{F}_D$. Hence any tensor of the form $\omega\otimes \alpha$ with $\omega$ exact and $\alpha\in [S(A_F),S(A_F)]$ belongs to $\mathcal{F}_D$. Now if we act on such a tensor with $g\otimes \beta\in \mathcal{C}^\infty(M)\otimes [S(A_F),S(A_F)]$, we obtain $(g\otimes \beta)\cdot (\omega\otimes \alpha)=g\omega\otimes[\beta,\alpha]$. Using sums of terms like this and the fact that $[S(A_F),S(A_F)]$ is perfect, we see that  $\Omega^1_M\otimes [S(A_F),S(A_F)]\subset\mathcal{F}_D$, and the theorem is proved.
\end{proof}

Let us apply theorem \ref{thfluct} to find the nature of  gauge fields in the case where $A_F=\bigoplus_{i=1}^kJ_i$ with $J_i=H_{n_i}(\mathbb{K})$ or $\mathrm{JSpin}(n_i)$ (see theorem \ref{ClassificationTheorem}).   Using the cube \eqref{liftdiagcompleted} and $\ker\pi=0$, we obtain that
\begin{equation}
[S(A_F),S(A_F)]=\bigoplus_{i=1}^k\mathrm{Der}_{\mathrm{Inn}}(J_i)
\end{equation}
If $J_i=H_{n_i}(\RR), H_{n_i}(\CC), H_{n_i}(\HH)$ or $\mathrm{JSpin(n_i)}$, then $\mathrm{Der}_{\mathrm{Inn}}(J_i)=so(n_i),su(n_i),sp(n_i)$ or $\mathrm{spin}(n_i)=so(n_i)$ respectively. This directly gives the nature of the gauge fields from the algebra $A_F$.  In particular we see that unimodularity is a natural feature of the  Jordan setting.

We now turn our attention to the general fluctuation space.

\begin{theorem}
Let $\mathcal{F}_\omega^{M\times F}$ be the general fluctuation space of the almost-associative background $\mathcal{B}_M\hat\otimes\mathcal{B}_F$. Then
\begin{equation}
\mathcal{F}_\omega^{M\times F}=\Omega^1_M\otimes [S(A_F),S(A_F)]\oplus \mathcal{C}^\infty(M,\mathcal{F}_\omega^F)
\end{equation}
where $\mathcal{F}_\omega^{ F}$ is the general fluctuation space of $\mathcal{B}_F$.
\end{theorem}
\begin{proof}
With the same notations as above, we let $a=f\otimes a_F\in A$ and $\omega=\omega_M\otimes b_F+g\otimes \omega_F\in \Omega^1_{M\times F}$. Then, using $f=f^o$ and $\omega_M^o=-\omega_M$ and suppressing $\pi$ for simplicity, we have
\begin{eqnarray}
[S_a,S_\omega]&=&\frac{1}{4} [a+a^o,\omega-\omega^o]\cr
&=&
\frac{1}{4} [f\otimes (a_F+a_F^o),\omega_M\otimes b_F+g\otimes \omega_F-(-\omega_M\otimes b_F^o+g\otimes \omega_F^o]\cr
&=&\frac{1}{4}f\omega_M\otimes [a_F+a_F^o,b_F+b_F^o]+fg\otimes [a_F+a_F^o,\omega_F-\omega_F^o]\cr
&=&f\omega_M\otimes [S_{a_F},S_{a_F}]+fg\otimes [S_{a_F},S_{\omega_F}]
\end{eqnarray}
The result follows.
\end{proof}

{Notice in particular that both the minimal and general gauge fluctuations will be the same. Differences may arise in the Higgs sector of a model however.
}

\section{Boyle-Farnsworth Model}\label{BFmodel}
\subsection{Definition of the model}
Let us consider the model with finite algebra $ A_F=\mathrm{JSpin}(2)\oplus H_2(\CC)\oplus H_3(\CC)\oplus\RR$ which has been proposed by Boyle and Farnsworth \cite{Boyle:2020}.  We will identify the elements of $\mathrm{JSpin}(2)$ with matrices of the form $\begin{pmatrix}
x&z^*\cr z&x
\end{pmatrix}$, with $x\in \RR$ and $z\in\CC$. The Hilbert space $H_F$  is the same as  for the SM, cf  \eqref{SMHilbertSpace}. The associative representation is 

\begin{equation}
\pi(\lambda,h,m,r)=[\lambda\otimes 1_4,h\otimes 1_4,r\oplus 1_2\otimes m,r\oplus 1_2\otimes m]\label{finitealgelem}
\end{equation}

where $[A,B,C,D]$ denotes a matrix which is block-diagonal with respect to the chirality space.  More generally we will write $a_{XY}$ with $X,Y=R,L,\bar R,\bar L$ for a matrix decomposed into chiral blocks.

\subsection{Gauge fields}

From section \ref{AAsec} we know that gauge fields will take values in some representation of $u(1)\oplus su(2)\oplus su(3)$ (since $u(1)\simeq so(2)$). To work out the precise representations, we first express the symmetrized action $S(A_F)$, which contains general elements that are sums of the  form:
\begin{eqnarray}
2S(\lambda)&=&[\lambda\otimes 1_4,0,\lambda\otimes 1_4,0]\cr
2S(h)&=&[0,h\otimes 1_4,0,h\otimes 1_4]\cr
2S(m)&=&[0\oplus 1_2\otimes m,0\oplus 1_2\otimes m,0\oplus 1_2\otimes m,0\oplus 1_2\otimes m]\cr
2S(r)&=&[r\oplus 0,r\oplus 0,r\oplus 0,r\oplus 0]
\end{eqnarray}
Because elements of different kinds commute with one another, and $S(r)$ commutes with everything, we therefore find that $[S(A_F),S(A_F)]$ is generated as a vector space by elements of the form 
\begin{eqnarray}
T(\lambda')&:=&[\lambda'\otimes 1_4,0,\lambda'\otimes 1_4,0]\cr
T(h')&:=&[0,h'\otimes 1_4,0,h'\otimes 1_4]\cr
T(m')&:=&[0\oplus 1_2\otimes m',0\oplus 1_2\otimes m',0\oplus 1_2\otimes m',0\oplus 1_2\otimes m']\label{gaugegen}
\end{eqnarray}
where edit  $\lambda'\in \RR\begin{pmatrix}
i&0\cr 0&-i
\end{pmatrix}=u(1)$, $h'\in [H_2(\CC),H_2(\CC)]=su(2)$, and $m'\in [H_3(\CC),H_3(\CC)]=su(3)$. 	
	
As noted by the authors in~\cite{Boyle:2020}, even though the Lie algebra is the correct one, the $u(1)$ charges corresponds to a linear combination of hypercharge and B-L.   %(Note the possibility of replacing the last $\RR$ factor which does not yield gauge fields with another $H_2(\RR)$. We would then obtain two $u(1)$-gauge fields acting independently on left and right particles.)
 In order to obtain the correct hyper-charges, the gauge symmetries would need to be extended to include the anomaly-free outer automorphisms of the representation. In this case one obtains the correct hypercharges, but the  model is extended by an additional gauged $B-L$ symmetry. We will not consider such an extension in this paper.

\subsection{Higgs sector}

 We recall that for the noncommutative standard model the Dirac operator takes the form:

\begin{equation}
D_F=\begin{pmatrix}
0&Y^\dagger&M^\dagger&0\cr Y&0&0&0\cr M&0&0&Y^T\cr 0&0&Y^*&0 
\end{pmatrix}\label{finiteDirac}
\end{equation}
where $Y=Y_\ell\oplus Y_q$, $Y_\ell=\begin{pmatrix}
Y_\nu&0\cr 0&Y_e
\end{pmatrix}$, $Y_q=\begin{pmatrix}
Y_u&0\cr 0&Y_d
\end{pmatrix}$, and $M=\begin{pmatrix}
m_\nu&0\cr 0&0
\end{pmatrix}\oplus 0$.

We will take this $D_F$ as our starting point.  To check that $C_1$ (resp. weak $C_1$) holds, we need only consider commutators of the form $[(a')^o,[D,a]]$ with $a,a'\in \pi(A_F)$ (resp. $a'\in [\pi(A_F),\pi(A_F)]$).  Consider first the case where $a'\in \pi(A_F)$. Using  the same notations as in \eqref{finitealgelem} for $a$ (and the same with primes for $a'$), we easily find that $[(a')^o,[D,a]]$ is a selfadjoint matrix with all blocks vanishing except the $(R,\bar R)$ and $(\bar R,R)$ ones, the latter being given by $(\lambda'-r')M(\lambda-r)$. Thus $C_1$ is not satisfied unless $M=0$,  in which case both blocks vanish. On the other hand, if $a'\in[\pi(A_F),\pi(A_F)]$, then $r'=0$ and $\lambda'=t\begin{pmatrix}
i&0\cr 0&-i
\end{pmatrix}$, $t\in\RR$. Such an element is easily seen to be in $\Omega^1_{D_F}$. In other words, $D_F$ is seen to satisfy weak $C_1$, while $C_1$ is only satisfied for $M=0$. We therefore see that stronger restrictions arise on the Higgs sector under $C_1$, than occur in the associative NCG SM.

Satisfied that this Dirac operator is compatible with weak $C_1$, let us next determine the minimal finite fluctuations. These are obtained by taking iterated commutators of elements \eqref{gaugegen} with $D_F$. As the matrix $D_F$ commutes with $T(m')$, however, we have only to focus on iterated commutators of  $D_F$ with elements of the form $T(h')$ and $T(\lambda')$. Beginning with $T(h')$ yields:
\begin{equation}
\Phi(q):=\begin{pmatrix}
0&Y(q)^\dagger&0&0\cr Y(q)&0&0&0\cr 0&0&0&Y(q)^T\cr 0&0&Y(q)^*&0 
\end{pmatrix}\label{defPhi}
\end{equation}
where $Y(q)=qY$, $q\in su(2)$ (a pure quaternion). Commuting with $T(\lambda')$  
  boils down to multiplying $q$ with $\begin{pmatrix}
 i&0\cr 0&-i
 \end{pmatrix}$ which is another quaternion, which gives nothing new. The Higgs sector is thus the same as in the Standard Model when we consider the minimal fluctuation space.

 Let us next look at the general fluctuation space. We begin by enforcing $C_1$, which sets $M=0$. We then need to determine the form of the Jordan module of finite $1$-forms.  Using hermiticity, we only need to consider the $(L,R)$-block. Now, if $a=[a_R,a_L,a_{\bar L},a_{\bar L}]$ and $b=[b_R,b_L,\ldots,\ldots]$, then the $(L,R)$-block of $a\circ [D,b]$ is
 \begin{equation}
 (a\circ [D_F,b])_{L,R}=a_LYb_R-a_Lb_LY+Yb_Ra_R-b_LYa_R
 \end{equation}
It is then easy to see that a general $1$-form $\omega$ will have a block $\omega_{LR}=\sum_{i,j}a_iYb_j$ where $a_i$ is in the associative $\RR$-algebra generated by $\mathrm{JSpin}(2)$ and $b_j$ is in the associative $\RR$-algebra generated by $H_2(\CC)$, that is $M_2(\CC)$. Now let us introduce the notation $\tilde Y:=Y\begin{pmatrix}
0&1\cr 1&0
\end{pmatrix}$. Using $Y\begin{pmatrix}
a&b\cr c&d
\end{pmatrix}=\begin{pmatrix}
a&0\cr 0&d
\end{pmatrix}Y+\begin{pmatrix}
b&0\cr 0&c
\end{pmatrix}\tilde{Y}$, we can rewrite $\omega_{LR}$ in the form $\sum x_iY+\tilde x_i\tilde Y$ where $x_i,\tilde{x_i}$ belong to the algebra generated by $\mathrm{JSpin}(2)$ and diagonal matrices, which is $M_2(\CC)$. 
Thus, $\Omega^1_F$ is 
  the set of matrices of the form 
\begin{equation}
\omega_g=\begin{pmatrix}
0&Z^\dagger g^\dagger &0&0\cr gZ&0&0&0\cr 0&0&0&0\cr 0&0&0&0 
\end{pmatrix}\label{finiteformsMzero}
\end{equation}
where $g\in M_2(\CC)$ and $Z$ is either equal to $Y$ or to $\tilde Y$.  From \eqref{genfluctC1}, the general fluctuations are thus of the form $[\pi(a),\omega_g]+J_F[\pi(a),\omega_g]J_F^{-1}$. Using the fact that $M_2(\CC)=\mathbb{H}\oplus i\mathbb{H}$, this yields four $SU(2)$-doublets $\Phi(q),\Phi(iq'),\tilde \Phi(p),\tilde\Phi(ip')$, where $q,q',p,p'$ are quaternions, $\Phi$ is defined by \eqref{defPhi} and $\tilde{\Phi}$ is the same as $\Phi$ with $\tilde{Y}$ replacing $Y$.  These fields will be independent iff the up and down components of $Y$ are. This example provides a proof that the general fluctuation space and the minimal one are in general different.

\section{The Pati-Salam model}\label{PSmodel}

\subsection{Definition of the model}
The finite Hilbert space $H_F$ is still the SM one, and $D_F$ still has the form \eqref{finiteDirac}. We take $A_F=A_{\rm PS}:=H_2(\CC)\oplus H_2(\CC)\oplus H_4(\CC)$ represented as

\begin{equation}
\pi(p,q,m)=[p\otimes 1_4,q\otimes 1_4,1_2\otimes m,1_2\otimes m]\otimes 1_N
\end{equation}

We have 

\begin{equation}
\pi(p,q,m)^o=[1_2\otimes m^*,1_2\otimes m^*,p^*\otimes 1_4,q^*\otimes 1_4]\otimes 1_N
\end{equation}

hence

\begin{equation}
2S(p,q,m)=[p\otimes 1_4+1_2\otimes m^*,q\otimes 1_4+1_2\otimes m^*,c.c.,c.c.]\otimes 1_N\label{SrepPS}
\end{equation}

\subsection{Gauge fields}
Taking the commutator of two elements like \eqref{SrepPS} we obtain that 

\begin{equation}
[S(A_F),S(A_F)]=su(2)_R\oplus su(2)_L\oplus su(4)\label{ScomSPS}
\end{equation}

 represented as

\begin{equation}
T(g_R,g_L,g)=[g_R\otimes 1_4+1_2\otimes g,g_L\otimes 1_4+1_2\otimes g,c.c,c.c]\otimes 1_N\label{genScomSPS}
\end{equation}

\subsection{Higgs sector}

The finite fluctuations are iterated commutators of elements of the form \eqref{genScomSPS} with $D_F$. More precisely, since the summands of \eqref{ScomSPS} commute among one another, a finite fluctuation has the form $T_1\cdot\ldots\cdot T_i\cdots S_1\cdot\ldots \cdot S_j\cdot R_1\cdot\ldots \cdot R_k\cdot D_F$, with $T_1,\ldots,T_i$ in $su(4)$, $S_1,\ldots,S_j\in su(2)_L$, and $R_1\cdots R_k\in su(2)_R$ (with the now usual convention that   $R_k$ operates first and $T_1$ last). The action of $R_1$ to $R_k$ on the $L-R$ sector of \eqref{finiteDirac} replaces $Y$ with $Y g_R^1\ldots g_R^k\otimes 1_4$ (up to an irrelevant sign). Then we act with $S_1,\ldots,S_j$  and get $g_L^j\ldots g_L^1\otimes 1_4 Y  g_R^1\ldots g_R^k\otimes 1_4$. Finally we act with $T_1,\ldots,T_i$, which in the $(L,R)$-sector amounts to taking the  commutators with $g_1,\ldots,g_i$. Hence the $L-R$ sector of a finite fluctuation contains a linear combination of elements of the form

\begin{equation}
1_2\otimes g_1\cdot \ldots g_i\cdot (p\otimes 1_4)Y(q\otimes 1_4)\label{PSHiggs}
\end{equation}

where $p$ and $q$ are products of elements of $su(2)$. Such products are just generic elements of $M_2(\CC)$.  Now let us write $Y_{\alpha,i}$ for the elements of the tensor $Y$, where $\alpha$ runs through $uu,ud,du,dd$, and $i$ through $\{\ell,r,g,b\}^2$. Hence, for each $i$, $Y_{.,i}$ is a $2\times 2$ matrix acting on $\CC^2_{\rm weak}$. If there are $N$ generations, $Y_{\alpha,i}$ will be a $N\times N$-matrix).  For each $i$, if $Y_{.,i}\not=0$, the products $pY_{.,i}q$ generate $M_2(\CC)$. Moreover, $M_4(\CC)$ decomposes as $\CC\oplus su(4)\oplus isu(4)$ under the adjoint action of $su(4)$. Hence \eqref{PSHiggs} will be a completely generic traceless element of $M_2(\CC)\otimes M_4(\CC)$ unless $Y_{\alpha,.}$ belongs to one of the submodules $\CC,su(4),isu(4),\CC\oplus su(4),\ldots$ for every $\alpha$. Given that $pY_{.,i}q$ generate $M_2(\CC)$ unless $Y_{.,i}$ vanishes, this is impossible for $Y\not=0$.
In conclusion a Higgs field takes values in  general traceless elements   $M_2(\CC)\otimes M_4(\CC)$. The  action of $su(2)_R$ is by multiplication on the right of $M_2(\CC)$, the action of $su(2)_L$ is by multiplication on the left of $M_2(\CC)$ and the action of $su(4)$ is the adjoint action on $M_4(\CC)$. This coincides with the results in \cite{Chamseddine:2013uq}.

\section{Beyond the special case}\label{secbeyond}

Our discussion has focussed on special Jordan coordinate algebras with associative specializations. We have looked specifically   at  real, even, special Jordan spectral triples (i.e. special Jordan triples equipped with real structure and grading operators), as these are the geometries that are of most interest when constructing physical theories. In this section we discuss the generalization to the exceptional setting, in which it is no longer possible to construct associative specializations. Our goal is not to give a complete and rigorous account, but rather to outline the most natural path forward that we see for constructing exceptional Jordan geometries. For this purpose we will focus on geometries without real structure, which will allow us to highlight  the primary distinctions between the special and exceptional cases without becoming tied down by interesting but tangential discussions concerning real structure, and the (weak) order conditions. 

Consider a special, Jordan, spectral triple without real structure  $\mathcal{T}=(A,H,\pi,D,\chi)$. A key feature of our construction is the extension of the associative  specialization $\pi$ of $A$ on $H$ to all expressions of degree 1 or less, such that
\begin{align}
\pi(a\circ b) &= \pi(a)\circ\pi(b)\\
\pi(a\circ \omega) &= \pi(a)\circ\pi(\omega)\label{assoiat1form}
\end{align}
for all $a,b\in A$, and $\omega\in \Omega^1 A$. Associative specializations readily allow for the representation  of general 1-forms once the representation of exact one forms is defined. Notice, however, that all other features of our construction actually derived from the properties of multiplicative specializations. Stated explicitly, an immediate consequence of theorem \ref{JacTh14} is that $\rho = \frac{1}{2}\pi$ satisfies all of the properties of a multiplicative specialization for expressions of degree 1 or less. The relationship between multiplicative specializations and Jordan modules then ensures that $(H,\rho)$ satisfies all the properties of a Jordan module over $A\oplus \Omega^1 A$ or equivalently that $\mathcal{T}$ can be viewed as a Jordan algebra (i.e. equations \eqref{equivmultrep} and \eqref{order0} are satisfied up to degree 1). It is the  Jordan module structure that ensures the form of the inner derivations of  $\mathcal{T}$ restricted to $H$, the form that inner fluctuations of $D$ take, and the Lie module structure of $\Omega^1_DA$.

The question of this section is how to extend the construction to accommodate exceptional Jordan coordinate algebras.  Exceptional Jordan algebras do not have  associative specializations,  but they do have multiplicative specializations, and as such it is natural to consider  Jordan spectral triples of the form  $\mathcal{T}=(A,H,\rho,D,\chi)$, where
\begin{enumerate}
	\item $A$ is a Jordan algebra with a JB algebra completion $\overline{A}$,
	\item $H$ is a Hilbert space,
	\item $\rho$ is a  faithful multiplicative  representation  of $A$,
	\item $D$ and $\chi$ satisfy all the usual properties,
	\item For all $a\in A$, $ [D,\rho(a)]$ is well-defined and bounded.
\end{enumerate}

Let us unpack this construction, beginning with the coordinate algebra and its representation on the Hilbert space  $(A,H,\rho)$. It should be stressed  that in the general Jordan setting we have in general that
\begin{align}
\rho_{a\circ b}\neq \rho_a\circ \rho_b,
\end{align}
for $a,b\in A$. Equation \eqref{assocrep} only holds for the special case in which one has an associative specialization, which is the key assumption that we are dropping in this section. In other words, the algebra generated by elements of the form $\rho_a$ for $a\in A$, equipped with the symmetrized product $\circ$, will be a special Jordan algebra that will not in  general be  isomorphic to $A$. Nonetheless, the properties of $\rho$ ensure that $B=A\oplus H$ is a Jordan algebra, and as such the form of its inner derivations are known. In particular, following equation \eqref{derivation0}, when restricted to  $H$ the inner derivations of $B$ are given by  $ \delta_{ab} = [\rho_a,\rho_b]$, for $a,b\in A$. Furthermore, following Equation~\eqref{order0}, we have
\begin{align}
[\delta_{ab},\rho_c] = \rho_{[b,c,a]}
\end{align}
for $c\in A$. The `lift' of the inner derivations of $A$ therefore remain exactly as in the special Jordan setting with no alteration.   At order zero our construction therefore ports readily to the exceptional setting. The same observation applies to     the minimal space $\cal{F}_D$ defined in section \ref{fluctuationsection}: it continues to make sense    as  the Lie module over  $\mathcal{G}(A)$ generated by elements of the form $[\delta,D]$, for $\delta\in \mathcal{G}(A)$.  Hence, a spectral triple over a non-special Jordan algebra together with its minimal fluctuation space could in principle be constructed.

In order to define an algebraic background and/or a general fluctuation space, one would need a general theory for the representation of Jordan $1$-forms. Let us say a few words about why such a construction is difficult. Assuming one is able to obtain a module of universal 1-forms with the universal	property~\eqref{UPomega1}\footnote{See~\cite{DV:2016} for an explicit construction of differential calculi over exceptional Jordan algebras.}, it would then be natural to consider the action of exact one forms  on $H$ following equation \eqref{da}:
\begin{align}
\rho_{d[a]}= [D,\rho_a],
\end{align}
where $a\in A$. Unfortunately,  the degree 1 extension of a multiplicative representation will have in general $\rho_{a\circ \omega}\neq \rho_a\circ \rho_\omega$, for $a\in A$, $\omega\in \Omega^1_d A$. It is therefore not possible to use our approach for special Jordan algebras to define the representation of general 1-forms, or to prove in general that the  map $a\mapsto [D,\rho(a)]$ is a derivation into $\rho(A)$.  This  is the primary obstruction to developing a completely general framework for Jordan coordinate algebras. Fortunately, while a general theory of Jordan $1$-forms seems to be difficult to reach, an approach tailored to the particular properties of the Albert algebra appears easier to obtain. Work in this direction is currently in progress, and is planned for an upcoming paper.

%Just as before this suggest the following definition:
%
%\begin{definition}
%Let $\mathcal{S}(A)$ be the Lie algebra generated by the operators $iS_a$, $a\in A$. Then the fluctuation space is the Lie module over $\mathcal{S}(A)$ generated by $D$.
%\end{definition}
%
%Let us recall the motivations for this definition:
%
%\begin{enumerate}
%\item It gives back what we already had in the special case. NO IT DOES NOT !
%\item Iterated commutators of things like $[S_a,S_b]$ with $D$ must be fluctuations because of \eqref{puregauge}.
%\item We cannot add $S_a$ to the Lie algebra because we need skew operators to make the commutator with $D$ self-adjoint. Hence the factor of $i$.
%\end{enumerate}
%
%It seems we cannot define 1-forms $\omega$ such that $\omega+\omega^o$ is a fluctuation, since with $S$ we only get symmetrized things. Maybe we don't need 1-forms. After all, fluctuations are what we really care about

%Now what about the order derivation $L_a$ ? Its natural extension is $S_a$. To close the cube we should have $\starad_{S_a}S_b=S_{L_ab}$, that is :
%
%$$S_aS_b+S_bS_a=S_{a\circ b}$$
%
%for all $b$. In the special case this yields 
    
\section{Conclusion}

%The spectral triple formulation of geometry shifts the focus from manifolds to the algebra of functions defined over them.   The advantage of this approach is that it allows access to new and interesting geometries corresponding to more general kinds of coordinate algebras. Dropping the assumption that a coordinate algebra should be commutative, for instance, gives rise to non-commutative geometry. More generally, one is able to consider noncommutative and nonassociatve geometries, although the nonassociative setting is much less well studied (see~\cite{} for previous work in this direction by one of the authors). An obvious starting point when exploring nonassociative geometries is to restrict attention to the commutative case. In this paper we explore the geometries corresponding to the most important class of commutative, nonassociative algebras, namely Jordan algebras, building upon previous work~\cite{Boyle:2014wba,Boyle:2020,carotenuto2019,DV:2016,Farnsworth:2014vva,Farnsworth:2020ozj,ShaneThesis,Farnsworth:2013nza}.
%

 In this paper we have proposed a definition for spectral triples and backgrounds based on Jordan coordinate algebras. Our treatment  focuses primarily on special Jordan algebras with associative representations. We have also outlined what we view to be the most natural path towards a generalization from the special Jordan case to the  exceptional setting. The key elements outlined in the extension of the spectral triple formalism to nonassociative geometries of Jordan type are:
\begin{enumerate}
	\item The definition of coordinate algebra representations, and the generalization of order conditions.
	\item The definition of internal symmetries, and their lifts to the representation space.
	\item The construction of differential forms and their representations.
	\item Inner fluctuations of the Dirac operator.	 
\end{enumerate}
Regarding the fluctuation space, we have shown that under the assumption of weak  $C_1$  the minimal space of fluctuated Dirac operators provides a  subspace of the configuration space which is automorphism invariant and thus can be used to define a consistent particle model. We have also explained the relationship of the general fluctuation space { to the fluctuation space given by Connes} in the associative setting.  An insight we have used is the construction of fluctuations as general derivation elements of degree 1. This was also seen when constructing nonassociative geometries of alternative type~\cite{Farnsworth:2013nza} (see equation 3.23 in that paper).  This observation should   allow one to extend the spectral formalism to general non-associative geometries, of which Jordan geometries are just one example.
 
We highlight a number of points in which Jordan geometries appear to do a better job of describing gauge theories than in the associative noncommutative setting:
\begin{enumerate}
\item As outlined in~\cite{Boyle:2020},  the Jordan setting appears to provide a more natural setting for describing Majorana fermions (and by extension provides a natural solution to the fermion doubling problem).
\item  The infinitesimal automorphisms of the coordinate algebra (inner derivations) are faithfully represented as operators on the Hilbert space.
\item { Because the infinitesimal automorphisms of the coordinate algebra are represented by commutators, unimodularity is naturally implemented.}
\item Jordan Banach are   more natural candidates for coordinatizing geometric spaces than real $C^*$ algebras.
\end{enumerate}

Despite these benefits there does appear to be one (potential) downside. In particular, the representation of $U(1)$ factors appears to be far more restrictive in the Jordan setting. This, however, might also be seen as a benefit from the perspective of construing realistic gauge theories, as it limits the allowable models that are able to be built (which is one of the key reasons for the interest in spectral geometry in the first place).

\end{document}